\documentclass[numberwithinsect,cleveref,a4paper]{lipics-v2021}
\nolinenumbers
\usepackage{cite}
\usepackage{amsmath,amssymb,amsfonts}
\usepackage{algorithmic}
\usepackage{graphicx}
\usepackage{textcomp}
\usepackage{xcolor}
\usepackage{comment}

\usepackage{amsthm}
\usepackage{todonotes}
\usepackage{graphicx}
\usepackage{xspace}
\usepackage{enumitem}
\usepackage{hyperref}
\usepackage{cleveref}
\usepackage{mathrsfs}

\usepackage[appendix=inline]{apxproof}

\usepackage{tikz}
\tikzstyle{every picture} = [>=latex]
\usetikzlibrary{calc,arrows,decorations.markings,quotes,shapes}
\usetikzlibrary{backgrounds}

\def\wreach#1#2#3{{\text{WReach}_{#1}[#2,\preceq_{#2},#3]}}

\def\ca#1{{\cal#1}}

\def\Val(#1){\mathop{\sf Val}(#1)}

\newtheorem{question}[theorem]{Question}
\theoremstyle{remark}

\makeatletter
\newtheorem{repeatclaim@}{Claim}

\newtheorem{repeatcase@}{Case}

\newtheorem{repeattheorem@}{Theorem}

\makeatother

\newtheorem{case}{Case}

\theoremstyle{definition}

\newenvironment{subproof}{%
  \begin{proof}[\sf Subproof]%
}{%
  \end{proof}%
}

\title{Transductions of Graph Classes Admitting Product~Structure
}
\titlerunning{Transductions of Classes Admitting Product~Structure}

\author{Petr Hlin{\v e}n\'y}{Masaryk University, Brno,
  Czech republic}{hlineny@fi.muni.cz}{https://orcid.org/0000-0003-2125-1514}{}
\author{Jan Jedelsk\'y}{Masaryk University, Brno,
  Czech republic}{484988@mail.muni.cz}{https://orcid.org/0000-0001-9585-2553}{}
  
\authorrunning{P.\ Hlin\v{e}n\'y and J.~Jedelsk\'y}
\Copyright{Petr Hlin\v{e}n\'y and Jan Jedelsk\'y}



\hideLIPIcs

\begin{document}

\maketitle

\begin{abstract}\,
In a quest to thoroughly understand the first-order transduction hierarchy of
hereditary graph classes, some questions in particular stand out;
such as, what properties hold for graph classes that are first-order
transductions of planar graphs (and of similar classes)?
When addressing this (so-far wide open) question, we turn to the concept of
a \emph{product structure} -- being a subgraph of the strong product of a
path and a graph of bounded tree-width, introduced by Dujmovi\'c et al.~[JACM~2020].
Namely, we prove that any graph class which is a first-order transduction of a
class admitting such product structure, up to perturbations also meets a
structural description generalizing the concept of a product structure in a
dense hereditary
way\,---\,the latter concept being introduced just recently by authors
under the name of \emph{$\ca H$-clique-width} [MFCS 2024].

Using this characterization, we show that the class of the 3D grids, as
well as a class of certain modifications of 2D grids,
are not first-order transducible from classes admitting a product structure,
and in particular not from the class of planar graphs.
\end{abstract}

\section{Introduction}
%%%%%%%%%%%%%%%%%%%%%%%%%%%%%%%%%%%%%%%%%%%%%%%%%%%%%%%%%%%%%%%%%%%%%%%

A recent research trend on the boundary between structural graph theory and
graph logic is to study \emph{first-order} (FO) trans\-ductions between (hereditary)
graph classes,
e.g.~\cite{DBLP:journals/corr/abs-2208-14412,DBLP:journals/corr/abs-2501.04166}.
This research is motivated by structural questions about dense
hereditary graph classes, and by the significant goal to resolve computational
complexity of the FO model checking problem on dense graph classes.

We refer to \Cref{sec:prelim} for the definitions.
Briefly, a simple first-order \emph{interpretation} is given by a binary FO formula
$\xi$ over graphs, and for a graph $G$ the result of the interpretation is
the graph $\xi(G)$ on the same vertex set and the edge relation (assumed
symmetric) determined by $uv\in E(\xi(G))$ $\iff$ $G\models\xi(u,v)$.
In a \emph{first-order transduction} $\tau$, one can additionally duplicate the vertices of
$G$ and assign arbitrary vertex colors before applying $\xi$, and then take
an induced subgraph of the result (in particular, $\tau(G)$ is actually a
set of graphs -- unlike $\xi(G)$).
These notions are naturally extended to graph classes; here $\tau(\ca C)$
denotes the class of all graphs obtained by applying $\tau$ to (colored)
graphs from a class~$\ca C$.

All transductions in this paper are first-order.
We say that a graph class $\ca D$ is \emph{transducible} from a class 
$\ca C$ if there is a transduction $\tau$ such that $\ca D\subseteq\tau(\ca C)$.
We interpret this relation as that $\ca D$ is not ``logically richer''
than~$\ca C$ (although $\ca D$ may be combinatorially a lot more complicated
class than~$\ca C$), and the aim is that if we understand and can handle the
class $\ca C$, we would be able to do the same with~$\ca D$.

There are two core and mutually related questions in the area;
for a class $\ca C$, describe structural properties satisfied by
(all) transductions of $\ca C$, and
for classes $\ca C,\ca D$, decide whether $\ca D$ is transducible from~$\ca C$.
For an illustration, let $\ca T$ be the class of trees,
then all transductions of $\ca T$ are of bounded clique-width, but not every
graph class of bounded clique-width is transducible from $\ca T$
(this would be true considering richer MSO transductions).
In fact, Ne\v{s}et\v{r}il et al.~\cite{DBLP:conf/soda/NesetrilMPRS21} 
showed that a graph class $\ca C$ is a first-order transduction of a class of
graphs of bounded tree-width, if and only if $\ca C$ is of bounded clique-width
and excludes some half-graph as a bi-induced subgraph
(we say that it has \emph{stable} edge relation).

Results of this kind can be useful also in the FO model checking problem.
Consider for example the following result of Gajarsk\'y et
al.~\cite{DBLP:journals/tocl/GajarskyHOLR20};
a graph class $\ca D$ is a transduction of a graph class of bounded maximum
degree, if and only if $\ca D$ is near-uniform.
Since, luckily, the claimed transduction can be efficiently computed, the
characterization \cite{DBLP:journals/tocl/GajarskyHOLR20} readily gives an
efficient FO model checking algorithm for each near-uniform graph class.
In general, when a graph class $\ca C$ has an efficient FO model checking
algorithm, and a class $\ca D$ is transducible from $\ca C$ (say,
via~$\tau$), there are still two obstacles for FO model checking on~$\ca D$\,:
First, for $G\in\ca D$ one has to be able to compute $G_1\in\ca C$ such that
$G\in\tau(G_1)$, and second, one also has to guess the coloring of $G_1$ used in
$\tau$ to obtain~$G$.

On the other hand, it is also interesting and useful to understand why a
graph class $\ca D$ is \emph{not} transducible from a graph class $\ca C$.
While on the positive side (`is transducible'), one can simply give a
transduction and guess a coloring, on the negative side the task is much
more difficult---one has to be able to find, for every transduction $\tau$,
a graph in $\ca D\setminus\tau(\ca C)$
and that is hard to argue.
In other words, in $\ca D$, one needs to understand the `obstacles' which are
not transducible from~$\ca C$.

For example, long paths are such obstructions for transducibility from a
class of bounded-height trees.
This is in fact a tight characterization;
Ossona de Mendez, Pilipczuk and Siebertz~\cite{DBLP:journals/ejc/MendezPS25}
proved that a graph class $\ca C$ is of bounded shrub-depth (that is, transducible
from a class of bounded-height trees), if and only if $\ca C$ does not
transduce the class of all paths.
Other, more complex examples of how understanding the obstructions for
transducibility is used to obtain strong results are contained, e.g., in
Bonnet et al.~\cite{DBLP:journals/jacm/BonnetGMSTT24} or
in Dreier, M{\"{a}}hlmann and Torunczyk~\cite{DBLP:conf/stoc/DreierMT24}.

However, except sporadic examples of success, we currently lack good tools
to prove non-transducibility.
A typical approach (to proving that a class $\ca D$ is not transducible
from~$\ca C$) suggests a structural property $\Pi$ satisfied by $\ca C$ and
preserved under transductions, and then finds graphs in $\ca D$ not
satisfying~$\Pi$.
The rather few examples of properties preserved under transductions include,
e.g.; monadic stability~\cite{SHELAHstabil}, 
bounded clique-width~\cite{DBLP:journals/mst/CourcelleMR00},
each value of shrub-depth~\cite{DBLP:journals/lmcs/GanianHNOM19},
bounded twin-width~\cite{DBLP:journals/jacm/BonnetKTW22}
and bounded flip-width~\cite{DBLP:conf/focs/Torunczyk23}.

\subsection{Our results}

Our research is inspired by the task to characterize classes which are
transducible from the class of planar graphs.
While planar graphs are usually perceived as combinatorially and
algorithmically ``nice'', their transductions (including, e.g., the
well-studied classes of $k$-planar and fan-planar graphs, or map graphs) 
can be combinatorially very complicated.
(For example, testing 1-planarity is NP-hard already for planar graphs with
one added edge~\cite{DBLP:journals/siamcomp/CabelloM13}.)
It is an open question whether the class of toroidal graphs (embeddable on the torus)
is transducible from planar graphs.

Not much was known prior to this paper about obstructions to transducibility 
from planar graphs except that the class of planar graphs is monadically stable and, more
recently, that such transductions are of bounded twin-width~\cite{DBLP:journals/jacm/BonnetKTW22}.
One can also prove, using known logical tools and the fact that planar
graphs are of bounded local tree-width, that classes transducible from
planar graphs are of bounded local clique-width up to perturbations.

On the other hand, several major breakthrough results about planar graphs in the combinatorics
domain, such as getting an upper bound on the queue number of planar graphs,
have been recently obtained using a new tool of a \emph{product
structure}~\cite{DBLP:journals/jacm/DujmovicJMMUW20}.
In the original basic setting, a class $\ca C$ \emph{admits a product structure} if, 
for some $k$, every graph in $\ca C$ is a subgraph of the strong product
$\boxtimes$ (\Cref{fig:strongprod}) of a path and a graph of tree-width~$\leq k$.
See \Cref{def:prodstruct}.
Planar and surface-embeddable graphs, for instance, do
admit product structure (\Cref{thm:admitprod}).

Our main result claims that transductions of classes admitting a
product structure essentially (up to bounded perturbations) follow a similar kind of
product structure (here briefly reformulated in view of \Cref{def:Hexpression} and \Cref{thm:hcwproduct}):
\begin{itemize}
\item (\Cref{thm:transd-admit-prod-str-paths})
  If $\ca C$ is a class admitting a product structure and $\tau$ is a
  transduction defining the class of simple graphs $\tau(\ca C)$,
  then every graph in $\tau(\ca C)$ is a
  bounded perturbation of an induced subgraph of the strong product
  of a path and a graph from a class of bounded clique-width.
\end{itemize}
We actually prove, in \Cref{thm:transd-admit-prod-str-bound-degree},
a generalized version in which paths (of the assumed product structure)
are replaced by graphs from any fixed class of bounded maximum degree.

We also give the converse direction of the main result, in a modification
consisting in strengthening the assumption of bounded clique-width to
bounded stable clique-width:
\begin{itemize}
\item (\Cref{thm:stablecw-back})
  Assume that every graph in a class $\ca C$ is a bounded perturbation of an
  induced subgraph of the strong product of a path
  and a graph from a class of bounded stable clique-width.
  Then, $\ca C$ is transducible from a class admitting a product structure.
\end{itemize}
Again, this is also generalized to any fixed class of bounded maximum degree
in place of the class of paths.

Finally, we demonstrate the strength of
\Cref{thm:transd-admit-prod-str-paths} by giving straightforward
proofs of the following consequences:
\begin{enumerate}[label={\Roman*.}]
\item\label{it:3Dclaim} (\Cref{cor:3Dgridsnot}) The class of all 3D grids, and
\item (\Cref{cor:2Dplusnot}) the class of certain non-local modifications of 2D grids
-- by adding a bunch of apex vertices adjacent to mutually distant nodes of the grid,
\end{enumerate}
are not transducible from any class admitting product structure, and in particular
not from the classes of planar graphs and graphs embeddable on a fixed surface.

At the same time and independently of us,
Gajarsk{\'{y}}, Pilipczuk and Pokr{\'{y}}vka
proved point \ref{it:3Dclaim} in \cite{DBLP:journals/corr/abs-2501-07558}.

\subsection{Proof outline}

The core of our results is in the following technical lemma:
\begin{itemize}
\item (\Cref{lem:transductions})
Let $Q$ be a graph of maximum degree~$d$, 
$Q^\circ_r$ denote the reflexive $r$-th power of~$Q$,
and $M$ be a graph of tree-width~$k$.
Every strongly $r$-local interpretation $\xi$
of any 2-colored spanning subgraph of
the product $Q\boxtimes M$ (cf.~\Cref{def:prodstruct})
has $\{Q^\circ_r\}$-clique-width (\Cref{def:Hexpression})
bounded in terms of $r,d,k$ and the rank of~$\xi$.
\end{itemize}

We note that, the requirement of $\ca Q$ having
bounded degree is necessary and best possible, since
Hlin\v{e}n\'y and Jedelsk\'y~\cite{DBLP:conf/mfcs/HlinenyJ24}
(in the extended arXiv version)
proved that already the class of strong products of two stars 
is monadically independent, that is, by suitably coloring
it, one can transduce all graphs from it.

\smallskip
We informally describe the proof idea of \Cref{lem:transductions}:

Let $G\subseteq Q \boxtimes M$, where $Q$ and $M$ are as above,
such that $V(G) = V(Q \boxtimes M)$.
Let $\xi$ be a strongly local formula.
The idea of the proof of \Cref{lem:transductions} 
is to recursively, in a bottom-up fashion, build
a $(Q, f(\ldots))$-expression (\Cref{def:Hexpression})
for the subgraphs $G'$ of $\xi(G)$ induced by $V(Q \boxtimes M')$,
where~$M'\subseteq M$ runs over subgraphs defined by subtrees in the
tree decomposition of~$M$.
The color of each vertex $v$ of $G'$ in the expression
encodes (via \Cref{def:exprcolors})
its type to all vertices outside $G'$ that
are reachable from $v$ by a path of bounded~length.

Using the information stored in the vertex colors
and $r$-locality of $\xi$, we decide the existence of 
an edge $uv$ in $\xi(G)$ depending only on the colors of $u$ and $v$
and intersection of $r$-balls around $u$ and $v$.
We prove correctness of this decision in
Claims~\ref{clm:alledges}~and~\ref{claim:no-extra-edges}.
Thus, the color of each vertex $v$ of $G'$
encodes its $r$-neighborhood outside $G'$.
Encoding the $r$-neighborhoods naively would
result in an unbounded number of colors. 
There are two core ideas enabling us to reuse colors
(as embedded in \Cref{def:exprcolors} and wrapped-up by \Cref{claim:count-colors}):
\begin{itemize}
  \item Strong locality of $\xi$ guarantees that
  we only ever create edges between vertices with intersecting
  $r$-balls, and so we can use the parameter graph $Q^\circ_r$
  to avoid creating edges at large $Q$-distance in the product~$Q\boxtimes M$.
  Consequently, we may ``remember'' the vertices of $Q$ by only their colors
  in a proper coloring of the $3r$-power of $Q$.
  \item We do not really remember all vertices of $V(G)$
  in the intersection of the $r$-neighborhoods of every pair
  of already-created vertices $u$ and $v$ of~$G'$.
  In fact, it suffices to remember the vertices
  of some $u$--$v$ separator of the union of
  the $r'$-neighborhoods of $u$ and of $v$ for some
  $r'$ depending on $r$ and the rank of $\xi$.
  We use weak colorings to bound the number of vertices 
  $m \in V(M) \setminus V(M')$ such that we have to
  encode some information concerning them in the color of a vertex $v$
  of~$V(G')$.
\end{itemize}

\subsection{Paper structure}

\Cref{sec:prelim} surveys the main related concepts and definitions.
\Cref{sec:transd-prod-str} proves \Cref{lem:transductions} and concludes
with the main results,
\Cref{thm:transd-admit-prod-str-bound-degree} and
\Cref{thm:transd-admit-prod-str-paths}.
\Cref{sec:stabilityback} shows the stable converse direction of the main
results, and \Cref{sec:3d-grids} proves the mentioned consequences on
non-transducibility of certain grids.
\Cref{sec:conclu} adds concluding remarks.

\section{Preliminaries}\label{sec:prelim}
%%%%%%%%%%%%%%%%%%%%%%%%%%%%%%%%%%%%%%%%%%%%%%%%%%%%%%%%%%%%%%%%%%%%%%%

\subsection{Graphs, decompositions, and weak colorings}
We write $[k]$ to denote the set $\{1,2,\ldots, k\}$.
In this paper, unless stated otherwise, \emph{graphs} are simple and finite,
\emph{loop graphs} additionally allow loops, and \emph{reflexive graphs}
are those with a loop incident to every vertex.
Let $V(G)$ denote the vertex set of a graph $G$ and $E(G)$ denote its edge
set, and let $G\subseteq H$ mean that $G$ is a subgraph of~$H$
and $G\subseteq_i H$ that $G$ is an induced subgraph of~$H$. 
A graph class $\ca C$ is \emph{hereditary} if $\ca C$ is closed under taking
induced subgraphs.

Let an \emph{$a\times b$-grid} (\emph{$a\times b\times c$-grid}) 
be the graph which is the Cartesian product% 
\footnote{Recall that in the Cartesian product of graphs, two tuples are
adjacent iff they agree in all coordinates except one in which they form an
edge (of the respective factor).}
of two paths on $a$ and $b$
vertices (three paths on $a$ and $b$ and $c$ vertices).
We also shortly call these graphs 2D and 3D grids.

For an integer $k\geq1$, let the \emph{$k$-th power} of a graph $G$ be the
graph on the same vertex set and edges joining exactly the pairs of distinct
vertices which are at distance at most $k$ in~$G$.

A graph $G$ is a {\em half-graph} of order $n$ if $G$ is a bipartite graph with the bipartition
$\{u_1,\ldots,u_n\}$ and $\{v_1,\ldots,v_n\}$, such that $u_iv_j\in E(G)$ if and only if~$i\leq j$.
A bipartite graph $G$ with a fixed bipartition $V(G)=A\cup B$ is a {\em bi-induced subgraph} of a graph $H$, 
if $G\subseteq H$ such that every edge of $H$ with one end in $A$ and the other end in $B$ is present in~$G$.

A \emph{tree decomposition} of a graph $G$ is a pair $(T, \ca X)$,
where $T$ is a tree and $\ca X = (X_t)_{t \in T}$ is a collection
of \emph{bags} $X_t \subseteq V(G)$ indexed by the nodes of $T$ satisfying
the following properties:
\begin{itemize}
  \item $\bigcup_{t \in V(T)} X_t = V(G)$
  \item For every edge $uv \in E(G)$, there is a bag
  $X_t$ where $t \in V(T)$ containing both $u$ and $v$.
  \item For every vertex $v \in V(G)$, the set of
  the nodes $t$ of $T$ such that $v \in X_t$ induces a 
  subtree of $T$.
  This is called the \emph{interpolation property}.
\end{itemize}
\emph{Width} of a tree decomposition $(T, \ca X)$ is
$\max_{t \in V(T)} |X_t| - 1$. The tree-width of a graph $G$ is the
minimum width over all tree decompositions of $G$.
We say that a tree decomposition $(T, \ca X)$ is \emph{smooth}
if, for every edge $t_1t_2 \in E(T)$ of the decomposition
tree $T$, we have that ${|X_{t_1} \setminus X_{t_2}| \le 1}$ and
${|X_{t_2} \setminus X_{t_1}| \le 1}$. Observe that
every tree decomposition can be modified to become
smooth without increasing its width.

A graph class $\ca C$ is said to be of \emph{bounded local tree-width}
if there exists an integer function $g$ such that, for every $G\in\ca C$,
$v\in V(G)$ and every $r$, the tree-width of the subgraph induced by the
neighborhood of $v$ in $G$ up to distance $r$ is at most~$g(r)$.

Given a pair of graphs $G$ and $H$, their \emph{strong product}
$G \boxtimes H$ is a graph on the vertex set
$V(G \boxtimes H) = V(G) \times V(H) = \{[g, h] |\> g \in V(G), h \in V(H)\}$
such that $\{[g_1, h_1],[g_2, h_2]\}$ is an edge of $G \boxtimes H$
if and only if one of the following is true;
\begin{itemize}
\item $g_1=g_2$ and $h_1h_2\in E(H)$, or
\item $g_1g_2\in E(G)$ and $h_1=h_2$, or
\item $g_1g_2\in E(G)$ and $h_1h_2 \in E(H)$.
\end{itemize}
See \Cref{fig:strongprod} for an illustration.
Notice that throughout the paper we use to refer to the vertices of the product $G \boxtimes H$
as to pairs $[g,h]$ where $g \in V(G)$ and $h \in V(H)$.

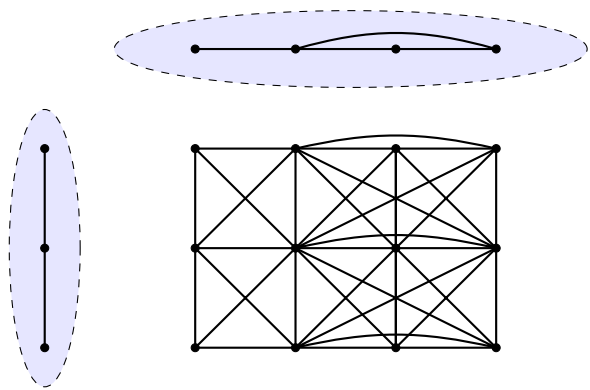
\begin{figure}[t]
$$
\begin{tikzpicture}[scale=0.66]
  \tikzstyle{every node}=[draw, black, shape=circle, minimum size=3pt,inner sep=0pt, fill=none]
  \tikzstyle{every path}=[draw, color=black!30!white]
  \draw (0,1) node[style=ellipse,dashed,fill=blue!10, inner xsep=3.3mm,inner ysep=13mm] {};
  \draw (6.1,5) node[style=ellipse,dashed,fill=blue!10,inner xsep=22mm,inner ysep=3.6mm] {};
  \foreach \x in {0,3,5,7,9}  \foreach \y in {5,3,1,-1} {
    \draw (\x,\y) node[fill=black] {};
  }
  \tikzstyle{every node}=[draw, white, fill=white, shape=circle, minimum size=3pt]
  \draw (0,5) node {};
  \tikzstyle{every path}=[draw, black, thick]
  \draw (0,-1)--(0,3);
  \draw (3,5)--(9,5) (5,5) to[bend left=16] (9,5);
  \draw (3,1)--(3,3) (5,1)--(5,3) (3,3)--(5,3) (3,1)--(5,1) (3,3)--(5,1) (3,1)--(5,3);
  \draw (3,-1)--(3,1) (5,-1)--(5,1) (3,-1)--(5,-1) (3,-1)--(5,1) (3,1)--(5,-1);
  \draw (7,1)--(7,3) (7,3)--(5,3) (7,1)--(5,1) (7,3)--(5,1) (7,1)--(5,3);
  \draw (7,-1)--(7,1) (7,-1)--(5,-1) (7,-1)--(5,1) (7,1)--(5,-1);
  \draw (7,1)--(7,3) (7,3)--(9,3) (7,1)--(9,1) (7,3)--(9,1) (7,1)--(9,3);
  \draw (7,-1)--(7,1) (7,-1)--(9,-1) (7,-1)--(9,1) (7,1)--(9,-1);
  \draw (9,1)--(9,3) (9,3) to[bend right=13] (5,3) (9,1) to[bend right=13] (5,1) (9,3)--(5,1) (9,1)--(5,3);
  \draw (9,-1)--(9,1) (9,-1) to[bend right=13] (5,-1) (9,-1)--(5,1) (9,1)--(5,-1);
\end{tikzpicture} 
\vspace*{-2ex}$$
\caption{Illustrating the strong product $\boxtimes$ of the two shaded graphs.}
\label{fig:strongprod}
\end{figure}

\begin{definition}[Product structure]\label{def:prodstruct}
For any class $\ca Q$ of bounded-degree graphs,
we say that a graph class $\ca C$ \emph{admits $\ca Q$-product structure} if
there is a constant $k$ (depending on $\ca C$), such that for every graph $G \in \ca C$
there is a graph $Q\in\ca Q$ and a graph $M$ of tree-width at most~$k$ such that
$G\subseteq Q \boxtimes M$.
Specially, we say that a graph class $\ca C$ \emph{admits product structure} if
$\ca C$ admits $\ca Q$-product structure where $\ca Q$ is the class of paths.
(In the latter case, we sometimes add an adjective \emph{classical} to
such product structure with paths.)
\end{definition}
Notably, the following hold:
\begin{theorem}[{%
\cite{DBLP:journals/jacm/DujmovicJMMUW20,DBLP:journals/combinatorics/UeckerdtWY22,%
DBLP:journals/dmtcs/DistelHHW22,DBLP:journals/jctb/DujmovicMW23,DBLP:journals/corr/abs-2001-08860%
}}]\label{thm:admitprod}
Graph classes admitting (the classical) product structure include, e.g.:
\begin{itemize}
\item[--] planar graphs, and graphs embeddable on a fixed surface,
\item[--] graphs drawable in the plane or a fixed surface with bounded
number of crossings per edge ($k$-planar graphs),
\item[--] apex-minor-free graphs.
\end{itemize}
\end{theorem}
Moreover, a mild extension of product structure (allowing for apex vertices
and clique-sums) holds for every proper minor-closed class
\cite{DBLP:journals/jacm/DujmovicJMMUW20}.

Finally, we recall the notion of weak coloring numbers.
Let $G$ be a graph and let $\preceq_G$ be a linear order of
its vertex set $V(G)$. We say that a vertex $u$
is \emph{weakly $r$-reachable (with respect to $\preceq_G$)}
from a vertex $v$
if there is a $u$--$v$ path $P$ of length at most $r$
such that, for all vertices $w \in V(P)$, we have $u \preceq_G w$. 
We denote by $\wreach{r}{G}{v}$
the set of all vertices weakly $r$-reachable from $u$.
The \emph{weak coloring number} of a graph $G$
is the minimum over all linear orders $\preceq_G$
of $\max_{v \in V(G)} \big|\wreach{r}{G}{v}\big|$.

\subsection{$\ca H$-clique-width}

We use the following new concept of Hlin\v{e}n\'y and 
Jedelsk\'y~\cite{DBLP:conf/mfcs/HlinenyJ24} extending the traditional
clique-width structural measure~\cite{DBLP:conf/gg/CourcelleER90}.
Let $H$ be a loop graph and let $k \ge 2$ be an integer.
Consider labeled graphs, where each vertex $v$ is assigned a label
$(h, c)$ such that $h \in V(H)$ is the \emph{parameter vertex} of $v$
and $c \in [k]$ is the \emph{color} of $v$.
\begin{definition}\label{def:Hexpression}
An \emph{$(H, k)$-expression} is an algebraic expression
composed of the following operations whose values are labeled graphs
as considered above:
\begin{enumerate}[label={\alph*)}]
  \item Given a \emph{parameter vertex} $h \in V(H)$ and a \emph{color}
  $c \in [k]$, create a one-vertex graph with the vertex labeled~$(h, c)$.
  \item Take the disjoint union of two labeled graphs.
  \item Given a pair of colors $c_1, c_2 \in [k]$, \emph{recolor}
  $c_1$ to $c_2$ in a labeled graph;
  that is, replace each label $(h, c_1)$ by label $(h, c_2)$.
  Notably, this operation only modifies the colors, it
  does not modify the parameter vertices.
  \item Given a pair of colors $c_1, c_2 \in [k]$, \emph{add edges}
  between colors $c_1$ and $c_2$; that is,
  make adjacent all pairs of vertices $u$ and $v$ such
  that the label of $u$ is $(h, c_1)$, the label of $v$ is $(h', c_2)$,
  and the parameter vertices of $u$ and $v$ are adjacent in $H$;
  $hh'\in E(H)$.
\end{enumerate}
\end{definition}
We refer to \Cref{fig:hcw-grid} for an illustration of \Cref{def:Hexpression},
and to \cite{DBLP:conf/mfcs/HlinenyJ24} for combinatorial properties of this notion.
In particular, observe that the mapping from the vertices of an expression
to their parameter vertices is a graph homomorphism to~$H$.

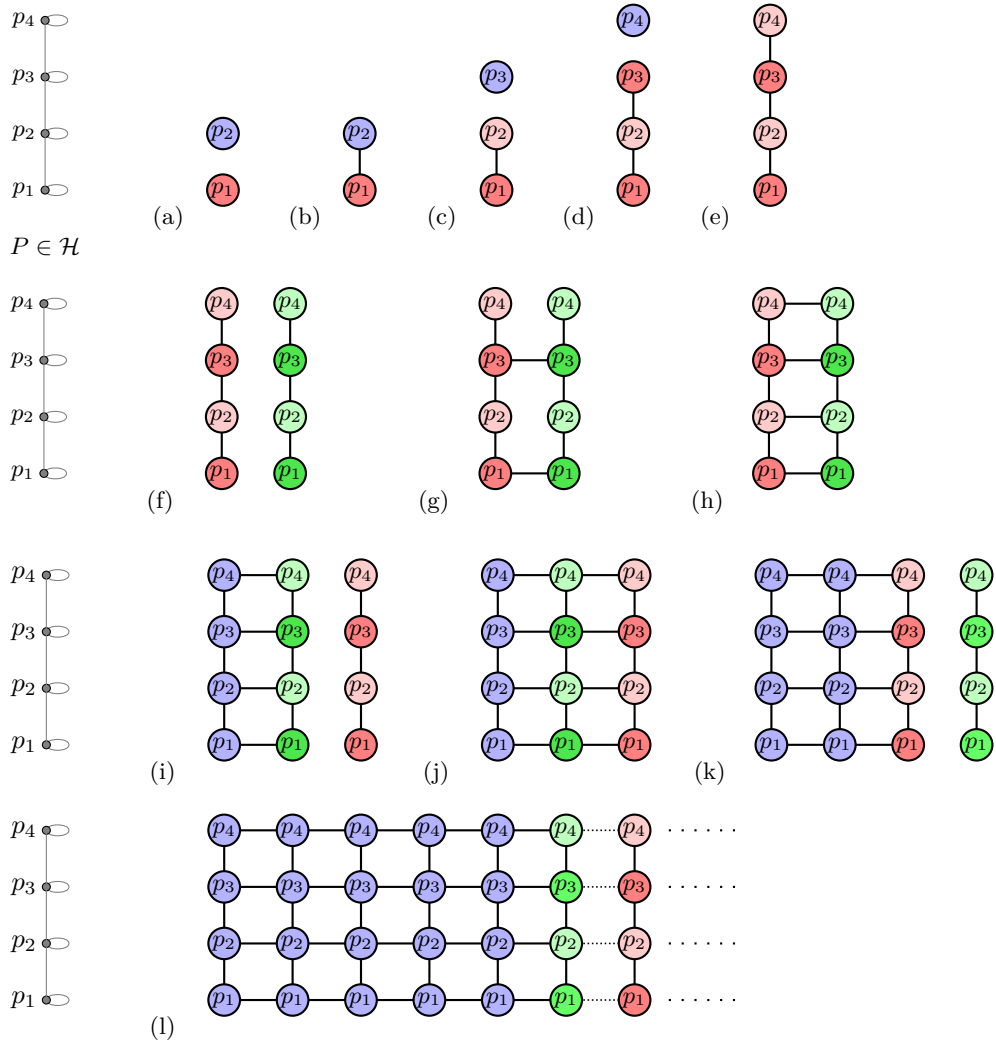
\begin{figure}[tp]
\def\xscal{0.9}%
\def\yscal{0.75}%
\def\lef{0.4}%
\begin{tikzpicture}[xscale=\xscal, yscale=\yscal]\small
  \tikzstyle{every node}=[draw, black, shape=circle, minimum size=3pt,inner sep=0.5pt, fill=gray]
  \tikzstyle{every path}=[draw, color=gray]
  \foreach \y in {4,3,2,1} {
    \draw (\lef,\y) node[label=left:$p_\y$] {} to[out=90,in=90] ++(0.33,0) to[out=270,in=270] ++(-0.33,0);
  } \draw (\lef,1)--(\lef,4);
  \node[draw=none,fill=none] at (\lef,0) {$P\in\ca H$};
  \tikzstyle{every path}=[draw, color=black, thick]\small
  \draw (3,1) node[fill=red!50!white] {$p_1$};
  \draw (3,2) node[fill=blue!30!white] {$p_2$};
  \node[draw=none,fill=none] at (2.2,0.5) {(a)};
  \draw (5,1)--(5,2);
  \draw (5,1) node[fill=red!50!white] {$p_1$};
  \draw (5,2) node[fill=blue!30!white] {$p_2$};
  \node[draw=none,fill=none] at (4.2,0.5) {(b)};
  \draw (7,1)--(7,2);
  \draw (7,1) node[fill=red!50!white] {$p_1$};
  \draw (7,2) node[fill=red!20!white] {$p_2$};
  \draw (7,3) node[fill=blue!30!white] {$p_3$};
  \node[draw=none,fill=none] at (6.2,0.5) {(c)};
  \draw (9,1)--(9,3);
  \draw (9,1) node[fill=red!50!white] {$p_1$};
  \draw (9,2) node[fill=red!20!white] {$p_2$};
  \draw (9,3) node[fill=red!50!white] {$p_3$};
  \draw (9,4) node[fill=blue!30!white] {$p_4$};
  \node[draw=none,fill=none] at (8.2,0.5) {(d)};
  \draw (11,1)--(11,4);
  \foreach \y in {4,2} \draw (11,\y) node[fill=red!20!white] {$p_\y$};
  \foreach \y in {3,1} \draw (11,\y) node[fill=red!50!white] {$p_\y$};
  \node[draw=none,fill=none] at (10.2,0.5) {(e)};
\end{tikzpicture} 

\begin{tikzpicture}[xscale=\xscal, yscale=\yscal]\small
  \tikzstyle{every node}=[draw, black, shape=circle, minimum size=3pt,inner sep=0.5pt, fill=gray]
  \tikzstyle{every path}=[draw, color=gray]
  \foreach \y in {4,3,2,1} {
    \draw (\lef,\y) node[label=left:$p_\y$] {} to[out=90,in=90] ++(0.33,0) to[out=270,in=270] ++(-0.33,0);
  } \draw (\lef,1)--(\lef,4);
  \tikzstyle{every path}=[draw, color=black, thick]\small
  \draw (3,1)--(3,4);
  \foreach \y in {4,2} \draw (3,\y) node[fill=red!20!white] {$p_\y$};
  \foreach \y in {3,1} \draw (3,\y) node[fill=red!50!white] {$p_\y$};
  \draw (4,1)--(4,4);
  \foreach \y in {4,2} \draw (4,\y) node[fill=green!25!white] {$p_\y$};
  \foreach \y in {3,1} \draw (4,\y) node[fill=green!60!lightgray] {$p_\y$};
  \node[draw=none,fill=none] at (2.1,0.5) {(f)};
  \foreach \y in {3,1} \draw (7,\y)--++(1,0);
  \draw (7,1)--(7,4);
  \foreach \y in {4,2} \draw (7,\y) node[fill=red!20!white] {$p_\y$};
  \foreach \y in {3,1} \draw (7,\y) node[fill=red!50!white] {$p_\y$};
  \draw (8,1)--(8,4);
  \foreach \y in {4,2} \draw (8,\y) node[fill=green!25!white] {$p_\y$};
  \foreach \y in {3,1} \draw (8,\y) node[fill=green!60!lightgray] {$p_\y$};
  \node[draw=none,fill=none] at (6.1,0.5) {(g)};
  \foreach \y in {4,3,2,1} \draw (11,\y)--++(1,0);
  \draw (11,1)--(11,4);
  \foreach \y in {4,2} \draw (11,\y) node[fill=red!20!white] {$p_\y$};
  \foreach \y in {3,1} \draw (11,\y) node[fill=red!50!white] {$p_\y$};
  \draw (12,1)--(12,4);
  \foreach \y in {4,2} \draw (12,\y) node[fill=green!25!white] {$p_\y$};
  \foreach \y in {3,1} \draw (12,\y) node[fill=green!60!lightgray] {$p_\y$};
  \node[draw=none,fill=none] at (10.1,0.5) {(h)};
\end{tikzpicture} 
\bigskip

\begin{tikzpicture}[xscale=\xscal, yscale=\yscal]
  \tikzstyle{every node}=[draw, black, shape=circle, minimum size=3pt,inner sep=0.5pt, fill=gray]
  \tikzstyle{every path}=[draw, color=gray]
  \foreach \y in {4,3,2,1} {
    \draw (\lef,\y) node[label=left:$p_\y$] {} to[out=90,in=90] ++(0.33,0) to[out=270,in=270] ++(-0.33,0);
  } \draw (\lef,1)--(\lef,4);
  \tikzstyle{every path}=[draw, color=black, thick]\small
  \foreach \y in {4,3,2,1} \draw (3,\y)--++(1,0);
  \draw (3,1)--(3,4);
  \foreach \y in {4,2} \draw (3,\y) node[fill=blue!30!white] {$p_\y$};
  \foreach \y in {3,1} \draw (3,\y) node[fill=blue!30!white] {$p_\y$};
  \draw (4,1)--(4,4);
  \foreach \y in {4,2} \draw (4,\y) node[fill=green!25!white] {$p_\y$};
  \foreach \y in {3,1} \draw (4,\y) node[fill=green!60!lightgray] {$p_\y$};
  \draw (5,1)--(5,4);
  \foreach \y in {4,2} \draw (5,\y) node[fill=red!20!white] {$p_\y$};
  \foreach \y in {3,1} \draw (5,\y) node[fill=red!50!white] {$p_\y$};
  \node[draw=none,fill=none] at (2.1,0.5) {(i)};
  \foreach \y in {4,3,2,1} \draw (7,\y)--++(2,0);
  \draw (7,1)--(7,4);
  \foreach \y in {4,2} \draw (7,\y) node[fill=blue!30!white] {$p_\y$};
  \foreach \y in {3,1} \draw (7,\y) node[fill=blue!30!white] {$p_\y$};
  \draw (8,1)--(8,4);
  \foreach \y in {4,2} \draw (8,\y) node[fill=green!25!white] {$p_\y$};
  \foreach \y in {3,1} \draw (8,\y) node[fill=green!60!lightgray] {$p_\y$};
  \draw (9,1)--(9,4);
  \foreach \y in {4,2} \draw (9,\y) node[fill=red!20!white] {$p_\y$};
  \foreach \y in {3,1} \draw (9,\y) node[fill=red!50!white] {$p_\y$};
  \node[draw=none,fill=none] at (6.1,0.5) {(j)};
  \foreach \y in {4,3,2,1} \draw (11,\y)--++(2,0);
  \draw (11,1)--(11,4);
  \foreach \y in {4,2} \draw (11,\y) node[fill=blue!30!white] {$p_\y$};
  \foreach \y in {3,1} \draw (11,\y) node[fill=blue!30!white] {$p_\y$};
  \draw (12,1)--(12,4);
  \foreach \y in {4,2} \draw (12,\y) node[fill=blue!30!white] {$p_\y$};
  \foreach \y in {3,1} \draw (12,\y) node[fill=blue!30!white] {$p_\y$};
  \draw (13,1)--(13,4);
  \foreach \y in {4,2} \draw (13,\y) node[fill=red!20!white] {$p_\y$};
  \foreach \y in {3,1} \draw (13,\y) node[fill=red!50!white] {$p_\y$};
  \draw (14,1)--(14,4);
  \foreach \y in {4,2} \draw (14,\y) node[fill=green!25!white] {$p_\y$};
  \foreach \y in {3,1} \draw (14,\y) node[fill=green!60!white] {$p_\y$};
  \node[draw=none,fill=none] at (10.1,0.5) {(k)};
\end{tikzpicture} 
\medskip

\begin{tikzpicture}[xscale=\xscal, yscale=\yscal]
  \tikzstyle{every node}=[draw, black, shape=circle, minimum size=3pt,inner sep=0.5pt, fill=gray]
  \tikzstyle{every path}=[draw, color=gray]
  \foreach \y in {4,3,2,1} {
    \draw (\lef,\y) node[label=left:$p_\y$] {} to[out=90,in=90] ++(0.33,0) to[out=270,in=270] ++(-0.33,0);
  } \draw (\lef,1)--(\lef,4);
  \tikzstyle{every path}=[draw, color=black, thick]\small
  \foreach \y in {4,3,2,1} \draw (3,\y)--++(5,0);
  \draw (3,1)--(3,4);
  \foreach \y in {4,2} \draw (3,\y) node[fill=blue!30!white] {$p_\y$};
  \foreach \y in {3,1} \draw (3,\y) node[fill=blue!30!white] {$p_\y$};
  \draw (4,1)--(4,4);
  \foreach \y in {4,2} \draw (4,\y) node[fill=blue!30!white] {$p_\y$};
  \foreach \y in {3,1} \draw (4,\y) node[fill=blue!30!white] {$p_\y$};
  \draw (5,1)--(5,4);
  \foreach \y in {4,2} \draw (5,\y) node[fill=blue!30!white] {$p_\y$};
  \foreach \y in {3,1} \draw (5,\y) node[fill=blue!30!white] {$p_\y$};
  \draw (6,1)--(6,4);
  \foreach \y in {4,2} \draw (6,\y) node[fill=blue!30!white] {$p_\y$};
  \foreach \y in {3,1} \draw (6,\y) node[fill=blue!30!white] {$p_\y$};
  \draw (7,1)--(7,4);
  \foreach \y in {4,2} \draw (7,\y) node[fill=blue!30!white] {$p_\y$};
  \foreach \y in {3,1} \draw (7,\y) node[fill=blue!30!white] {$p_\y$};
  \foreach \y in {4,3,2,1} \draw[semithick,densely dotted] (8,\y)--++(1,0);
  \draw (8,1)--(8,4);
  \foreach \y in {4,2} \draw (8,\y) node[fill=green!25!white] {$p_\y$};
  \foreach \y in {3,1} \draw (8,\y) node[fill=green!60!white] {$p_\y$};
  \draw (9,1)--(9,4);
  \foreach \y in {4,2} \draw (9,\y) node[fill=red!20!white] {$p_\y$};
  \foreach \y in {3,1} \draw (9,\y) node[fill=red!50!white] {$p_\y$};
  \foreach \y in {4,3,2,1} \draw[thick,loosely dotted] (9.5,\y)--++(1,0);
  \node[draw=none,fill=none] at (2.1,0.5) {(l)};
\end{tikzpicture} 
\caption{
An informal illustration of a $(P,5)$-expression making the $a\times b\,$-grid,
where $P$ is a $b$-vertex reflexive path (here $b=4$).
Each vertex in the expression is labeled $(p_i,c)$ where $p_i\in V(P)$ is
the parameter vertex inscribed in the vertex, and $c\in[5]$ is the color,
here depicted as light/dark red or green, or blue.
In the first phase, see (a)--(e), we create the vertical edges of the grid,
independently for each column (with alternating shades of red or green).
In the second phase, see (f)--(l), we union the columns and add by the
definition precisely the required edges in the rows -- between dark red 
and dark green, and separately between light red and light green colors.
}
\label{fig:hcw-grid}
\end{figure}

Given a (fixed) class of loop graphs $\ca H$, we define the \emph{\mbox{$\ca H$-clique-width}}
of a (simple) graph $G$ as the smallest integer $\ell$ such that (a coloring
of) the graph $G$ is the value of an $(H, \ell)$-expression for some $H\in\ca H$.

The concept is closely related to the aforementioned product
structure via the following rather straightforward statement:
\begin{theorem}[Hlin\v{e}n\'y and                 
Jedelsk\'y~\cite{DBLP:conf/mfcs/HlinenyJ24}]\label{thm:hcwproduct}
Let $\ca H$ be a family of reflexive graphs and $\ca H'$ be the
family obtained from $\ca H$ by removing all loops of all graphs.
For every $\ell\geq2$;
a simple graph $G$ is of $\ca H$-clique-width at most $\ell$, 
if and only if $G$ is an induced subgraph
of the strong product $H'\boxtimes M$ where $H'\in\ca H'$ and $M$ is a simple graph of clique-width at most~$\ell$.
\end{theorem}

Considering an $(H, \ell)$-expression $\phi$, the \emph{deparameterization of $\phi$} 
is an expression $\phi_1$ obtained as follows:
Let $K_1^\circ$ denote a single vertex graph formed by a vertex $v$ with a loop.
Then $\phi_1$ is the $(K_1^\circ, \ell)$-expression constructed from $\phi$ by replacing
every parameter vertex with~$v$.
(Note that ``deparameterized $\ca H$-clique-width'' is exactly the usual \emph{clique-width}.)
Furthermore, we say that an $(H, \ell)$-expression $\phi$ is
\emph{$k$-stable} if the value of the deparameterization $\phi_1$ of $\phi$ does
not contain a half-graph of order $k$ as a bi-induced subgraph.

\subsection{Relational structures and logic}
A \emph{relational structure} is an $\ell$-tuple $(V, R_1, \ldots, R_\ell)$,
where~$V$ is the universe and $R_i \subseteq V^{a_i}$ is a relation
of arity~$a_i$. The \emph{signature} of such relational structure
is the tuple~$(R_1, \ldots, R_\ell)$ of symbols together with their arities.
We view $c$-colored graphs $G=(V, E, C_1, \ldots, C_c)$
as relational structures over a signature with the binary edge relation
and unary relations for each~color. 

In this article, unless we say otherwise,
we only consider first-order logic, thus
all our formulas are first-order.
In order to simplify our notation, we assume that
free variables of each formula are in the set
$\{x_i | i \in \mathbb{N}\}$ and bound variables
are in the set $\{y_i | i \in \mathbb{N}\}$.
We write $\phi(x_1, x_2, \ldots, x_k)$ to denote that
the set of free variables of $\phi$ is a subset of
$\{x_1, x_2, \ldots, x_k\}$. Then, given a relational
structure $G$ and a $k$-tuple
$\bar{v}=(v_1, v_2, \ldots, v_k)$
of elements of $V(G)$, we write
$G \models \phi(\bar{v})$ to denote
$G, (x_1 \mapsto v_1), (x_2 \mapsto v_2), \ldots,
(x_k \mapsto v_k) \models \phi$,
where $x_k \mapsto v_k$ denotes that
the value of $x_k$ is $v_k$.

We often simply say that a formula $\phi$
has \emph{arity} $k$ to denote
$\phi(x_1, x_2, \ldots, x_k)$.
Given a formula $\phi(x_1, \ldots, x_k)$ and indices
$i_1, i_2, \ldots, i_k$, we denote by
$\phi[x_{1} \to x_{i_1}, x_{2} \to x_{i_2}, \ldots, x_{k} \to x_{i_k}]$
the formula obtained from $\phi$ by renaming $x_{j}$ to $x_{i_j}$
for all $1 \le j \le k$. We note that, due to
the assumption from the beginning of this paragraph,
the renaming only affects free variables, it does not
affect any bound variable.

Given a signature $\sigma=(R_1, \ldots, R_\ell)$
and formulas $\xi_1, \ldots, \xi_{\ell'}$ together with their arities
$a_1, \ldots, a_{\ell'}$ and a formula $\omega$ of arity
$1$, using predicates from $\sigma$,
the \emph{$(\xi_1, \ldots, \xi_{\ell'},\> \omega)$-interpretation} of a
relational structure $G$ over $\sigma$ is the relational
structure $H=(V^H, R_1^H, \ldots, R_{\ell'}^H)$, where
$v \in V^H \subseteq V^G$ iff $G \models \omega(v)$, and each
$R_i^H$ has the same arity $a_i$ as $\xi_i$ and
$(v_1, \ldots, v_{a_i}) \in R_i^H$ iff $G \models \xi_i$. 
An interpretation is \emph{first-order} if all
involved formulas are first-order.

Given a formula $\xi$ and a relational structure
$G$, we denote by $\xi(G)$ the result of the $(\xi, \top)$-interpretation
of $G$, where $\top$ is just tautology (of arity 1).
A $k$-\emph{copy} of a relational
structure $G=(V^G, R_1^G, \ldots, R_\ell^G)$ is the relational structure
$H=(V^G\times[k], R_1^H, \ldots, R_\ell^H,\> SC, SV)$
obtained from $G$ by taking the disjoint union of $k$ copies of $G$
and adding binary relations
$SV = \{([u, i], [u, j]) |\> u \in V^G, i, j \in [k]\}$
and $SC = \{([u, i], [v, i]) |\> u, v \in V^G, i \in [k]\}$.
\mbox{A \emph{$k$-coloring}} of a relational structure $G=(V^G, R_1^G, \ldots, R_\ell^G)$
is the operation resulting in the class of relational structures
$H=(V^G,$ $R_1^G, \ldots, R_\ell^G,$ $C_1, \ldots, C_k)$
where $C_i$ are interpreted as arbitrary unary relations over $V^G$.

The operations of interpretation, copying, and coloring naturally extend
to classes of relational structures by applying them to each member
of the class.
A \emph{(first-order) transduction} is any composition of
copying, coloring, and a (first-order) interpretation.
Given a class $\ca C$ and a transduction
$\tau$, we denote by $\tau(\ca C)$ the class obtained
from $\ca C$ by applying $\tau$.
In this paper, unless stated otherwise, we consider transductions of
classes of colored graphs producing classes of graphs.

A transduction is \emph{non-copying} if the copying operation is not used.
We say that a class $\ca C$ \emph{does not need copying}
if, for each $k$, there is a non-copying transduction $\tau_k$
such that $k\ca C \subseteq \tau_k(\ca C)$, where $k\ca C$
denotes the $k$-copy of $\ca C$. 
Similarly, we say that $\ca C$ \emph{does not need coloring}%
\footnote{See \Cref{claim:leaf-closed-no-color-copy} for an example of such class.}
if, for each $k$, there is an interpretation $\xi_k$
such that $\ca C^c \subseteq \xi_k(\ca C)$
for every $k$-coloring $\ca C^c$~of~$\ca C$.

Given an integer $r$, a graph $G$, and a tuple of its vertices~$\bar{v}$,
the \emph{$r$-ball} $B_r(G, \bar{v})$ around $\bar{v}$ in $G$ is
the subgraph of $G$ induced by the vertices at distance at most $r$
from at least one of the vertices of $\bar{v}$.
A formula $\xi$ of arity $a$ is \emph{$r$-local} if,
for every graph $G$ and every $a$-tuple of its vertices $\bar{v}$,
we have that $G \models \xi(\bar{v})$ iff $B_r(G, \bar{v}) \models \xi(\bar{v})$.
An $r$-local formula $\xi$ of arity $a$ is \emph{strongly $r$-local} if,
for every graph $G$ and every $a$-tuple of its vertices $\bar{v}$,~
$G \models \xi(\bar{v})$ implies that the maximum pairwise distance between
vertices of $\bar{v}$ is at most $r$. 

We say that a non-copying transduction is \emph{immersive} if the formulas 
defining the involved interpretations are strongly local.
We say that a transduction $\tau$ is \emph{subsumed} by
a transduction $\upsilon$ if for every class $\ca C$
we get $\tau(\ca C) \subseteq \upsilon(\ca C)$. 

We say that a graph $G$ is a $k$-\emph{perturbation} of a graph $H$
if $V(G) = V(H)$ and there is a partition $\ca P=(P_1, P_2, \ldots P_k)$
of $V(G) = V(H)$ into $k$ parts and a symmetric relation $\sim$ on $[k]$
such that the following holds:
For every pair of distinct vertices $u\in P_i$ and $v\in P_j$ 
where $i,j\in[k]$ (possibly $i=j$),
$uv \in E(G)$ if and only if $(uv \in E(H) \iff i\sim j)$.
The notion naturally extends to graph classes; a class $\ca C$
is a $k$-perturbation of a class $\ca D$ if each graph of $\ca C$
is a $k$-perturbation of some graph of $\ca D$.
Perturbations are sometimes also called flips in the literature.

\begin{theorem}[Ne\v{s}et\v{r}il, Ossona de Mendez, and Siebertz
\cite{DBLP:conf/csl/NesetrilMS22}]\label{thm:pertimmersive}
Any non-copying first-order transduction is subsumed by the 
composition of an immersive transduction and a perturbation.
\end{theorem}

In the coming proof, we employ the 
notion of a \emph{$q$-type} of a $k$-tuple $\bar{v}$ of vertices
of a (colored) graph $G$, that is, the set $tp_{q}(G, \bar{u})$
of all first-order formulas $\xi(x_1, \ldots, x_k)$ of quantifier rank
at most $q$ satisfying $G \models \xi(\bar{v})$.
We note that the number of distinct $q$-type of $k$-tuple over a fixed
signature is bounded by a function of $q$ and $k$.
We refer the reader to Chapter~I.3 of Achim Blumensath's
open access book \cite{achim-book} for more details about
finiteness of types.

We often need to work with types of sets rather that tuples.
To that end, we say that
a $k$-tuple $\bar{u}=(u_1, \ldots, u_k)$ is \emph{$\preceq$-ordered},
if $u_1\preceq u_2\preceq\ldots\preceq u_k$ where $\preceq$ is a linear order.
Observe that $\preceq$-ordered tuple $\bar{u}$ is uniquely determined by
its set of vertices $\{u_1, \ldots, u_k\}$ and their multiplicity.

Consider tuples of vertices of a graph $G$. We say that a tuple
$\bar{u}$ $d$-separates tuples $\bar{v}_1, \ldots, \bar{v}_k$
if, for every pair $i \neq j \in [k]$,
every path $P$ of length at most $d$ between a vertex of
$\bar{v}_i$ and a vertex of $\bar{v}_j$ contains a vertex of $\bar{u}$.
Dreier \cite{DBLP:conf/lics/Dreier21} proved a useful strengthening
of Feferman-Vaught theorem showing how to recover the $q$-type
of a tuple from $g(q)$-types of its subtuples:

\begin{theorem}[Dreier \cite{DBLP:conf/lics/Dreier21}]\label{thm:merge-types}
  There exists a function $g(q, \ell)$ such that for all
  labeled graphs $G$, every $q, \ell \in\mathbb N$, and all
  tuples $\bar{u}, \bar{v}_{1}, \ldots,
  \bar{v}_{k}$ of vertices from $G$ such that 
  $\bar{u}$ $4^q$-se\-parates $\bar{v}_{1},
  \ldots, \bar{v}_{k}$ and
  $|\bar{u}|+|\bar{v}_{1}|+\ldots+|\bar{v}_{k}| \le \ell$,
  the type $tp_q(G, \bar{u} \bar{v}_{1} \ldots \bar{v}_{k})$
  depends only on the types
  $tp_{g(q,\ell)}(G, \bar{u} \bar{v}_{1}),$ 
  $\ldots, tp_{g(q,\ell)}(G, \bar{u} \bar{v}_{k})$.
  Furthermore, $tp_q(G, \bar{u} \bar{v}_{1} \ldots
  \bar{v}_{k})$ can be computed from $tp_{g(q,\ell)}(G, \bar{u} \bar{v}_{1}), 
  \ldots,$ $tp_{g(q,\ell)}(G, \bar{u} \bar{v}_{k})$.
\end{theorem}
\smallskip

\section{Transductions of classes admitting product structure}
%%%%%%%%%%%%%%%%%%%%%%%%%%%%%%%%%%%%%%%%%%%%%%%%%%%%%%%%%%%%%%%%%%%%%%%
\label{sec:transd-prod-str}

The following technical lemma constitutes the core of the proofs of our main
results. We formulate it for 2-colored graphs (e.g., having black and white
vertices) with the further intention to smoothly handle induced subgraphs.

\begin{lemma}\label{lem:transductions}
  There is a function $f: \mathbb{N}^4 \to \mathbb{N}$
  such that the following holds:

  Let $\xi$ be a strongly $r$-local first-order binary formula of
  quantifier rank $q$ over the signature of 2-colored graphs.
  Let $\ca Q$ be a class of graphs of degree at most $d$.
  For any $Q \in \ca Q$
  let $G$ be a 2-colored graph, such that $V(G) = V(Q \boxtimes M)$
  and $E(G) \subseteq E(Q \boxtimes M)$, where 
  $M$ has tree-width at most $k$.

  Then, $\xi(G)$ has
  $\{Q^\circ_r\}$-clique-width at~most~$f(r, q, d, k)$,
  where $Q^\circ_r$ denotes the reflexive closure of the $r$-th power of~$Q$.
\end{lemma}

\begin{proof}
  We postpone the definition of $f$ until the end
  of this proof.
  In the coming proof we will only consider the case that $r \ge 4^q$.
  This assumption is possible thanks to the following simple claim:
  
  \begin{claim}\label{claim:r-th-power-only}
    Let $\phi$ be a 
    $(Q^\circ_{\max\{r, 4^q\}}, f(r, q, d, k))$-expression
    valued $\xi(G)$. Then, the value of $\phi$ as a
    $(Q^\circ_{r}, f(r, q, d, k))$-expression equals $\xi(G)$.
  \end{claim}
    \begin{subproof}
      We denote by $G_r$ the value of $\phi$
      as a $(Q^\circ_{r}, f(r, q, d, k))$-expression;
      this is sound since $Q^\circ_{\max\{r, 4^q\}}$ and $Q^\circ_{r}$ have equal vertex set.
      Since $Q^\circ_{r}$ is a subgraph of $Q^\circ_{\max\{r, 4^q\}}$,
      we get that $G_r$ is a subgraph of $\xi(G)$.
      Suppose for a contradiction that there is a pair of
      vertices $u=[p_u, h_u], v=[p_v, h_v] \in V(G_r) = V(\xi(G)) = V(G)$
      such that $uv \in E(\xi(G)) \setminus E(G_r)$.
      Then, the distance in $Q$ between $p_u$ and $p_v$
      is strictly more than $r$ (and at most $\max\{r, 4^q\}$).
      Hence, the distance between $u$ and $v$ in $G$ is strictly
      greater than $r$, too, and so the edge $uv$ contradicts
      the assumption of strong $r$-locality of $\xi$.
    \end{subproof}

  Let $\ca T = (T, \ca X=(X_t)_{t \in V(T)})$ be
  a smooth tree decomposition of $M$ of width $k$ rooted in
  $r_0 \in V(T)$ such that the root bag $X_{r_0}=\emptyset$ is empty,
  every node of $T$ has at most two descendants,
  and for each leaf $t \in V(T) \setminus \{r_0\}$
  the bag $X_t$ contains exactly one vertex.
  We observe that we can modify any tree decomposition
  to satisfy our requirements.

  We denote by $t^\uparrow$ the parent node of $t$ in $T$.
  It follows from $X_{r_0}=\emptyset$ and from the interpolation
  property that for every vertex $h \in V(M)$
  there is a unique node $t_h \in V(T)$ such that
  $h \in X_{t_h} \setminus X_{t_h^\uparrow}$.
  We say that $t_h$ is the \textit{forget node} of $h$.
  Let $\preceq_T$ denote the leaves-to-root preorder of $T$.
  That is, $r_0 \preceq_T t$ for all $t \in V(T)$.
  We fix an arbitrary linear order $\preceq_M$
  satisfying that, for all $h_1, h_2 \in V(M)$,
  if $t_{h_1} \preceq_T t_{h_2}$ then $h_1 \preceq_M h_2$.

  We proceed in a way similar to the proof \cite[Theorem~5.2]{DBLP:conf/mfcs/HlinenyJ24}.
  However, roughly, instead of remembering adjacency to the not-yet-forgotten
  neighbors, we remember the $q'$-type w.r. to the not-yet-forgotten
  vertices of the weak reachability sets for some constant $q'$ depending
  only on $r$, $q$, $d$, and $k$.

  For a node $t \in V(T)$, let $X^+_t\subseteq V(M)$ denote the union
  of $X_s$ where $s$ ranges over $t$ and all descendants of~$t$.
  Let $Y_t=X^+_t \setminus X_{t^\uparrow}$ denote the vertices of
  $M$ which occur {\em only} in the bags of $t$ and its descendants.
  For the root $r_0$ of $T$, let specially $Y_{r_0}=X^+_{r_0}=V(M)$.
  Observe that all neighbors of a vertex $h \in Y_t$ in
  $V(M) \setminus Y_t$ must belong to the set $X_t \setminus Y_t$,
  by the interpolation property of a tree decomposition.

  \begin{claim}\label{claim:r-separate-Y_t-from-outside}
    Any path $P\subseteq M$ of length at most $r$ between
    vertices $h_1 \in Y_t$ and $h_2 \in V(M) \setminus Y_t$
    intersects the set
    $\big(\wreach{r}{M}{h_1} \cap X_t\big) \setminus Y_t$.
  \end{claim}
  \begin{subproof}
    It follows from the interpolation property and
    the previous observation that there is
    a vertex $h \in (X_t \setminus Y_t) \cap V(P)$.
    Let $h$ be such vertex on the path $P$ closest to~$h_1$,
    and let $P'$ be the subpath of $P$ between $h_1$ and $h$.
    By the choices of $h$, we have $V(P') \setminus \{h\} \subseteq Y_t$,
    and so $h \preceq_M h'$ for all vertices $h' \in V(P')$.
    Hence, $h \in \wreach{r}{M}{h_1}$.
  \end{subproof}

  Grohe et al. \cite{DBLP:conf/wg/GroheKRSS15} proved that the
  weak $r$-coloring number of any graph $J$ is at most
  $wcol_r(J) \le {r + tw(J) \choose tw(J)}$.
  Notably, our $\preceq_M$ is the order used in their proof
  to witness that $M$ has its weak $r$-coloring number at most
  ${r + k \choose k}$. That is, for every $h \in V(M)$,
  we get that $\big|\wreach{r}{M}{h}\big| \le {r + k \choose k}$.
  \medskip

  We are approaching the core of the proof of \Cref{lem:transductions}.
  For an integer $\ell \ge 1$, we denote by $Q_\ell$ the
  $\ell$-th power of~$Q$, and
  by $Q^\circ_\ell$ the reflexive closure of $Q_\ell$.
  We fix a proper coloring
  $s_Q: V(Q) \to [d^{3r} + 1]$ of $Q_{3r}$.
  Such coloring exists since $Q$ has degree at most $d$,
  and so $Q_{3r}$ has degree at most $d^{3r}$.
  
  {\em Invariant.}
  We are going to build a $(Q^\circ_r, f(r, q, d, k))$-expression
  (cf.~\Cref{def:Hexpression})
  $\phi_t$ for each node $t \in V(T)$ using structural induction on $T$;
  maintaining the properties that
  \begin{enumerate}[label={(\roman*)}]
  \item the value of $\phi_t$ is the 
  subgraph $\xi(G)[V(Q) \times Y_t]$ of $\xi(G)$ induced by
  $V(Q) \times Y_t$, where
  \item the parameter vertex of each vertex $v=[p, h] \in V(Q) \times Y_t$
  is $p \in V(Q^\circ_r) = V(Q)$, and that
  \item the color $\hat c_{v,t}$ of $v$ is as prescribed by \Cref{def:exprcolors}.
  \end{enumerate}

  Before giving the definition of the color $\hat c_{v,t}$, we recall
  the statement of \Cref{thm:merge-types}, and set
  \begin{equation}
    g^{\max{}} := \max\{q, r, \max_{0 \le \,i\, \le 2\left(d^{3r}+1\right)\cdot{r + k \choose k}} g(q, i)\}
  ,\end{equation}
  where $g$ is the function from \Cref{thm:merge-types}.

  \begin{definition}[color~$\hat c_{v,t}$]\label{def:exprcolors}
    Fix an arbitrary linear order $\preceq_Q$ on $V(Q)$
    satisfying that, for all $p_1,p_2 \in V(Q)$,
    $s_Q(p_1) \le s_Q(p_2)$ implies $p_1 \preceq_Q p_2$.
    Let $\preceq_G$ be the linear order on $V(G) = V(Q)\times V(M)$
    obtained by ordering the vertices lexicographically 
    according to $\preceq_M$ and $\preceq_Q$, that is,
    $[p_1, h_1] \preceq_G [p_2, h_2]$ if and only if
    $(h_1 \preceq_M h_2 \land h_1 \neq h_2) \lor (p_1 \preceq_Q p_2 \land h_1 = h_2)$.
    Denote by $G^s$ the $2\!\left(d^{3r}+1\right)$-colored
    graph obtained from $G$ by coloring each vertex
    $v=[p, h]$ by the color $(c^G_v, s_Q(p))$,
    where $c^G_v$ is the given color of $v$ in 2-colored~$G$.

    We define the \emph{initial color} $\hat c_v$ of a vertex
    $v=[p_v, h_v] \in V(Q) \times V(M)$ as a {partial}(!) function
    \begin{equation}
      \hat c_v: 2^{[d^{3r}+1] \times \wreach{r}{M}{h_v}} \to \mbox{\it Types}
    ,\end{equation}
    where $\hat c_v(W)$ for $W\subseteq[d^{3r}+1] \times \wreach{r}{M}{h_v}$
    is partially defined as follows:
    \begin{itemize}\item
    Let $Z_{v, W} := \big\{u\!=\![p_u, h_u] \in V(G)|\> p_u \in V(B_r(Q, p_v))$
    $\wedge (s_Q(p_u), h_u) \in W\big\}$.
    Let $\bar{w}_{v, W}$ be the unique $\preceq_G$-ordered tuple
    containing each vertex of $Z_{v, W}$ exactly once, and let
    $\bar{w}^+_{v, W}:=(\bar{w}_{v, W}, v)$ denote this tuple after
    appending $v$ itself at the end.
    If $\{(s_Q(p), h) |$ $[p, h] \in Z_{v, W}\}=W$,%
\footnote{As a brief explanation;
  each member of $W$ determines (in the set $Z_{v, W}$)
  at most one vertex of $G$ at close distance 
  to $v$ (since $s_Q$ is a proper coloring of $Q_{3r}$).
  However, such vertex $u$ as determined by $(s_Q(p_u), h_u)\in W$
  may not exist in the graph~$G$.
  In order to properly capture the relevant vertices $Z_{v, W}$ of $G$ with respect to~$W$,
  we define $\hat c_v(W)$ only when the ``completeness'' condition
  $\{(s_Q(p), h) |$ $[p, h] \in Z_{v, W}\}=W$ is satisfied.
}
    then we let $\hat c_v(W)$ be the $g^{\max{}}$-type of $\bar{w}^+_{v,W}$ in $G^s$,
    and otherwise we leave $\hat c_v(W)$ undefined.
    \end{itemize}
    Subsequently, we define the \emph{running color} $\hat c_{v,t}$
    of a vertex $v=[p_v, h_v] \in V(Q) \times V(M)$ at a node $t\in V(T)$
    as the restriction of the initial color $\hat c_v$ of $v$ (above) to the sets
    ${W \subseteq \left[d^{3r}+1\right] \times \big(\wreach{r}{M}{h_v} \setminus Y_t\big)}$.
  \end{definition}

  Having defined the colors $\hat c_{c,t}$ (to be maintained in our construction),
  we will gradually describe the construction of the sought expressions $\phi_t$
  on the vertex set $V(Q) \times Y_t$ for $t\in V(T)$.
  We begin by independently constructing their subexpressions $\psi_{h}$
  valued the subgraph $\xi(G)[V(Q) \times \{h\}]$ for each vertex $h \in V(M)$
  in \Cref{claim:copy-of-Q}.
  \Cref{fig:hcw-grid} can be used as a ``lightweight'' (very simplified)
  illustration of this construction -- steps (a)--(e) in \Cref{fig:hcw-grid}
  illustrate building the expressions $\psi_h$, % in \Cref{claim:copy-of-Q},
  while steps (f)--(l) in the picture correspond to the subsequent
  construction of the expressions $\phi_t$ below.

  \begin{claim}\label{claim:copy-of-Q}
    There is a function $f': \mathbb{N}^4 \to \mathbb{N}$
    such that, for each vertex $h \in V(M)$, there is a 
    $(Q^\circ_r, f'(r, q, d, k))$-expression $\psi_{h}$
    valued $\xi(G)[V(Q) \times \{h\}]$,
    such that each vertex $v = [p_v, h] \in V(Q) \times \{h\}$
    has its initial color $\hat c_v$ and parameter vertex $p_v$.
  \end{claim}
  \begin{subproof}
    For each nonempty subset $Z \subseteq V(Q)$, by induction on $|Z|$,
    we build a $(Q^\circ_r, f(r, q, d, k))$-expression $\psi_{h}^Z$
    valued $\xi(G)[Z \times \{h\}]$, such that
    each vertex $v = [p_v, h] \in Z\times \{h\}$
    has color $(s_Q(p_v),\hat c_v)$ and parameter vertex $p_v$.

    In the base case, that is $|Z| = 1$, the expression
    $\psi_{h}^Z$ is trivial since we only create a single labeled vertex without edges.
    
    Assume $|Z| \ge 2$ and choose a vertex $v=[p_v, h]$ where $p_v\in Z$.
    We apply our induction assumption to obtain the
    expression $\psi_{h}^{Z\setminus\{p_v\}}$ and denote by $K$ its value.
    We make~$K+v$ by taking a disjoint union of
    $\psi_{h}^{Z\setminus\{p_v\}}$ and an expression creating
    $v$ with parameter vertex $p_v$ and a special color~\textsf{NewVertex}.
    Then, for every $u=[p_u, h] \in ({Z\setminus\{p_v\}})\times\{h\}$
    such that $uv \in E(\xi(G))$,
    we apply the operation of adding edges (cf.~\Cref{def:Hexpression}.d) between colors
    \textsf{NewVertex} and~$(s_Q(p_u),\hat c_u)$.

    A crucial observation is that, by doing so, we create only one edge for
    each such choice of~$u$, that is, the edge $uv$.
    Indeed, first note that since $\xi$ is strongly $r$-local, the distance
    between $u$ and $v$ in $G$, as well as the distance
    between $p_u$ and $p_v$ in $Q$, is at most $r$, and so $p_up_v\in E(Q_r^\circ)$.
    Let $x=[p_x, h]$ be any vertex of color equal to
    $(s_Q(p_u),\hat c_u)$ whose parameter vertex $p_x$ is
    adjacent to $p_v$ in $Q^\circ_{r}$.
    Then the distance in $Q$ between $p_x$ and $p_v$ is also at most $r$,
    and hence between $p_x$ and $p_u$ at most $2r$ by the triangle inequality.
    Since $s_Q$ is a proper coloring of $Q_{3r}$ and $s_Q(p_x)=s_Q(p_u)$,
    we conclude $p_x=p_u$, and consequently, $x=u$.

    Therefore, we have created an expression valued $\xi(G)[Z \times \{h\}]$,
    and now we recolor \textsf{NewVertex} to $(s_Q(p_v),\hat c_v)$
    to finish the desired expression $\psi_{h}^{Z}$.
    Finally, we obtain $\psi_{h}$ from $\psi_{h}^{V(Q)}$
    by recoloring each vertex from any color $(a, c)$ to color $c$.
    
    As the last task, we confirm that the number of colors used in the expression
    $\psi_{h}$ is bounded by a function ($f'$) of
    $r$, $q$, $d$, and $k$, because the size of the codomain of $s_Q$
    is at most $d^{3r}+1$, and unions of both domains and codomains
    of the functions $\hat c_u$ over all vertices
    $u \in V(Q) \times \{h\}$ have bounded size. 
    Indeed, all the domains are
    subsets of $\left[d^{3r}+1\right]\times\wreach{r}{M}{h}$,
    and so their union is of bounded size.
    The number of possible types in the codomains of functions
    $\hat c_u$ is bounded, since the number of colors is bounded,
    the quantifier rank is bounded, and the sizes of the tuples are bounded.
  \end{subproof}

  We now build the sought expression $\phi_t$. Assume
  that, for each child $t'$ of $t$, we have inductively built expressions
  $\phi_{t'}$ satisfying the aforementioned invariants (i), (ii), and (iii).
  We also shortly denote by $V(\phi)$ the vertex set of the value of~$\phi$.
  Then, we proceed with the following steps:
  \begin{enumerate}[label={\arabic*.}]\parskip1pt
    \item\label{step:leaf}
    If $t$ is a leaf, then $X_t=\{h\}$ for some node $h \in V(M)$.
    So, we create $\phi_t$ from $\psi_{h}$
    by recoloring each vertex $w$ from its initial color $\hat c_w$ to
    its running color $\hat c_{w,t}$. 

    \item We further assume that $t$ is not a leaf,
    and denote by $t_i$ for $i=1$ or $i\in\{1,2\}$ the children of $t$ in $T$.
    Let $\phi_{t_i}^\uparrow$ be the expression obtained
    from $\phi_{t_i}$ by recoloring each color $c$ to
    $(i, c)$ in $\phi_{t_i}^\uparrow$.
    If $t$ has only one child, we set $\phi_t^1:=\phi_{t_1}^\uparrow$.
    Otherwise, $t$ has exactly two children $t_1$ and $t_2$,
    and we create $\phi_t^1$ by taking the
    disjoint union of $\phi_{t_1}^\uparrow$~and~$\phi_{t_2}^\uparrow$.

    \item If $Y_t \setminus (Y_{t_1} \cup Y_{t_2}) = \emptyset$,
    then we set $\phi_t^2 := \phi_t^1$.
    Otherwise, it follows from smoothness of the decomposition of $M$
    that $Y_t \setminus (Y_{t_1} \cup Y_{t_2}) = \{h\}$ for some $h \in V(M)$.
    We create an expression $\psi_{h}^\uparrow$
    from $\psi_{h}$ by recoloring each color $c$ to $(3, c)$,
    and $\phi_t^2$ by taking the disjoint union of $\phi_t^1$
    and~$\psi_{h}^\uparrow$.

    \item\label{step:add-edges}
    We create $\phi_t^3$ from $\phi_t^2$ by adding
    edges between the following pairs of colors:
    \begin{enumerate}[label={(\Alph*)}]
      \item\label{step:add-edges-between-subtrees}
      If $t$ has two children, then we do the following:
      For every vertex pair $u \in V(Q)\times Y_{t_1}=V(\phi_{t_1})$ and
      $v \in V(Q)\times Y_{t_2}=V(\phi_{t_2})$ such that $uv \in E(\xi(G))$,
      we apply the operation of adding edges between
      colors $(1,\hat c_{u,t_1})$ and $(2,\hat c_{v,t_2})$
      (cf.~\Cref{def:Hexpression}.d).
      \item\label{step:add-edges-between-new-and-old}
      If $Y_t \setminus (Y_{t_1} \cup Y_{t_2}) \neq \emptyset$,
      that is, $Y_t \setminus (Y_{t_1} \cup Y_{t_2})= \{h\}$,
      then we do the following:
      For every $i\in\{1,2\}$ and every pair $u \in V(\phi_{t_i})$
      and $w \in V(Q)\times\{h\}=V(\psi_{h})$ such that $uw \in E(\xi(G))$,
      we apply the operation of adding edges between
      colors $(i,\hat c_{u,t_i})$ and $(3,\hat c_w)$.
    \end{enumerate}

    \item\label{step:final-recoloring}
    Finally, we make $\phi_t$ from $\phi_t^3$
    by performing the following recoloring operations to satisfy invariant (iii):
    \begin{itemize}
      \item For $i=1,2$, replace each color $(i,\hat c_{u,t_i})$
      by $\hat c_{u, t}$.
      \item Replace each color $(3,\hat c_{w})$ by $\hat c_{w, t}$.
    \end{itemize}
  \end{enumerate}

  We note that the recoloring operations of
  the steps~\ref{step:leaf}~and~\ref{step:final-recoloring} are well-defined,
  because the way we restrict the function $\hat c_{u,t_i}$
  to obtain $\hat c_{u, t}$ (resp. $\hat c_{w}$ to obtain $\hat c_{w, t}$)
  depends only on $t$.
  
  We denote by $G^t$ the value of $\phi_t$.
  It is clear that $V(G^t) = V(Q) \times Y_t$,
  and that the colors and the parameter vertices of $\phi_t$
  satisfy our invariants (ii) and (iii).
  We continue with:

  \begin{claim}\label{clm:alledges}
    $E\big(\xi(G)[V(Q) \times Y_t]\big) \subseteq E(G^t)$
  \end{claim}
  \begin{subproof}
    Suppose otherwise, and let
    $uv \in E(\xi(G)[V(Q) \times Y_t]) \setminus E(G^t)$
    be an edge we have not created.
    If $u,v \in V(\phi_{t_1})$ or $u,v \in V(\phi_{t_2})$,
    then we obtain a contradiction with the induction
    assumption. Similarly, if $u, v \in V(\psi_h)$,
    then we contradict \Cref{claim:copy-of-Q}.
    Hence, $uv$ is between two of the three subgraphs in the previous construction.
    It follows from strong $r$-locality of $\xi$
    that the distance between $u$ and $v$ in $G$
    is at most $r$, and so is the distance between their
    parameter vertices in $Q$, and
    the parameter vertices are adjacent in $Q^\circ_r$.
    Hence, we have created the edge $uv$ in
    step~\ref{step:add-edges} above. %, contradiction.
  \end{subproof}

  \begin{claim}\label{claim:no-extra-edges}
    $E(G^t) \subseteq E\big(\xi(G)[V(Q) \times Y_t]\big)$
  \end{claim}
  \begin{subproof}
    Suppose otherwise for a contradiction.
    Let $xy \in E(G^t) \setminus E(\xi(G)[V(Q) \times Y_t])$
    be any unwanted edge of $G^t$.
    As in \Cref{clm:alledges}, we must have added
    the edge $xy$ in step~\ref{step:add-edges}.
    Then, there is an edge $uv \in E(\xi(G)[V(Q) \times Y_t])$
    such that $x$ and $u$ as well as $y$ and $v$ have the same
    color. Let $[p_x, h_x]=x$, $[p_y, h_y]=y$,
    $[p_u, h_u]=u$, and $[p_v, h_v]=v$.

    It follows from strong $r$-locality of $\xi$
    that the distance between $u$ and $v$
    in $G$ it at most $r$, and so the distance
    between $h_u$ and $h_v$ in $M$ is at most $r$.
    Furthermore, we have created both edges $xy$ and $uv$, 
    and so $p_up_v\in E(Q_r^\circ)$ and $p_xp_y\in E(Q_r^\circ)$
    and, consequently, the $Q$-distance within each of the pairs
    $p_u,p_v$ and $p_x,p_y$ is at most $r$.

    In a nutshell,
    we are going to find an $r$-separator and use \Cref{thm:merge-types} to
    derive a contradiction. We denote by $\tilde c_x$, $\tilde c_y$,
    $\tilde c_u$, and $\tilde c_v$
\smallskip
    the underlying (initial or running) color of vertices $x$, $y$, $u$, and $v$
    when creating the edge~$xy$, and recall that
    these colors are partial functions
    $2^{[d^{3r}+1] \times \wreach{r}{M}{h_v}}$ $\to \mbox{\it Types}$.
    For both pairs $(a,b)\in\{(x,u),(y,v)\}$, we will find a set
    $W_{a,b}$ satisfying the following three conditions:
    \begin{itemize}
      \item[(I)] $W_{a,b}$ is the \emph{unique} inclusion-maximal subset of 
      $[d^{3r}+1] \times V(M)$ such that both values
      $\tilde c_{a}(W_{a,b})$ and
      $\tilde c_{b}(W_{a,b})$ are defined.
      \item[(II)] Both $\preceq_G$-ordered tuples $\bar{w}_{a, W_{a,b}}$ and
      $\bar{w}_{b, W_{a,b}}$ contain the same set of vertices $Z'_{a,b}$.
      \item[(III)] The set $Z'_{a,b}$ is an $r$-separator of $a$ and $b$.
    \end{itemize}

    We distinguish two cases, depending on whenever the edge $xy$ was
    created in step~\ref{step:add-edges}\ref{step:add-edges-between-subtrees} or
    step~\ref{step:add-edges}\ref{step:add-edges-between-new-and-old}.
    \begin{case}\label{xcase1}
      The edge $xy$ has been added in step~\ref{step:add-edges}\ref{step:add-edges-between-new-and-old}.
    \end{case}
    Without loss of generality, we assume that $x, u \in V(\phi_{t_i})$ where $i\in\{1,2\}$,
    and $y, v \in V(\psi_h)$ where $h=h_y=h_v$.
    Then $\tilde c_{x} = \hat c_{x, t_i}$, $\tilde c_{y} = \hat c_{y}$,
    $\tilde c_{u} = \hat c_{u, t_i}$, and $\tilde c_{v} = \hat c_{v}$.

    Recalling \Cref{def:exprcolors}, we observe that $x$ and $u$ (resp., $y$ and $v$)
    have the same type in the colored graph $G^s$,
    because $\hat c_{x, t_i}(\emptyset) =\hat c_{u, t_i}(\emptyset)$
    (resp., $\hat c_y(\emptyset) =\hat c_v(\emptyset)$)
    is always defined, and the value of $\hat c_z(\emptyset)=\hat c_{z, t'}(\emptyset)$
    is the type of $z$ for all vertices $z \in V(G)$ and all nodes $t' \in V(T)$.
    In particular, since each vertex color in $G^s$ contains the color of
    its parameter vertex, we conclude that $s_Q(p_x) = s_Q(p_u)$ and $s_Q(p_y)=s_Q(p_v)$.

    \medskip
    We continue the proof with three technical subclaims.
    For any $z=[p_z, h_z] \in V(\phi_{t_i})$,
    let us shortly denote
    $Z^1_{z}:=V(B_r(Q,p_z))\times\big(\wreach{r}{M}{h_z}\setminus Y_{t_i}\big)$.
    Similarly, for any $z=[p_z, h_z] \in V(\psi_{h})$, we shortly
    denote $Z^2_{z}:=V(B_r(Q,p_z))\times\wreach{r}{M}{h_z}$.

    Given any vertices $a=[p_a, h_a] \in V(\phi_{t_i})$ and
    $b=[p_b, h_b] \in V(\psi_{h})$,
    we denote by $Z'_{a,b}$ the set $Z'_{a,b}:=Z^1_a\cap Z^2_b$.
    Let $W_{a,b} := \{(s_Q(p),h) |\> (p,h) \in Z'_{a,b}\}$, and
    let $D_{a,b} \subseteq \mbox{$[d^{3r}+1]$} \times V(M)$ be an arbitrary
    inclusion-maximal set such that both values $\hat c_{a, t_i}(D_{a,b})$ and
    $\hat c_{b}(D_{a,b})$ are defined.

    In this case, conditions (I), (II), and (III) are fulfilled in order by the following three claims:
    \begin{claim}\label{xclm:D=W}
      Assume that distance between $p_a$ and $p_b$ is at most $r$ in $Q$.
      Then $D_{a,b} = W_{a,b}$, and so $D_{a,b}$ is uniquely determined.
    \end{claim}
    \begin{subproof}
        Let $(c, h') \in D_{a,b}$ be arbitrary.
        Then, the set $Z_{a, D_{a,b}}$ (cf.~\Cref{def:exprcolors})
        contains a vertex $[p^a, h']$ such that
        $c = s_Q(p^a)$ and $p^a \in V(B_r(Q, p_a))$.
        Furthermore, \mbox{$h' \in \wreach{r}{M}{h_a}\setminus Y_{t_i}$},
        because $\hat c_{a, t_i}$ is the restriction of $\hat c_a$
        to subsets of $[d^{3r}+1]\times(\wreach{r}{M}{h_a} \setminus Y_{t_i})$.
        So, $[p^a, h'] \in Z^1_a$.
  
        Similarly, the set $Z_{b, D_{a,b}}$ 
        contains a vertex $[p^b, h']$ such that
        $c = s_Q(p^b)$ and $p^b \in V(B_r(Q, p_b))$.
        Furthermore, $h' \in \wreach{r}{M}{h_b}$ since the domain
        of $\hat c_{b}$ contains only subsets of $[d^{3r}+1]\times\wreach{r}{M}{h_b}$.
        So, $[p^b, h'] \in Z^2_b$.

        The distance in $Q$ between $p_a$ and $p_b$ is at most $r$,
        and since $s_Q$ is a proper coloring of $Q_{3r}$,
        we get that $p^a = p^b$. Hence, $[p^a, h]=[p^b, h] \in Z'_{a,b}$,
        so $(c, h) = (s_Q(p^a), h) \in W_{a, b}$.

        Altogether, $D_{a,b} \subseteq W_{a, b}$.
      We observe that both $\hat c_a(W_{a,b})$ and
      $\hat c_b(W_{a,b})$ are defined as witnessed by the vertices in $Z'_{a,b}$.
      Hence, $W_{a,b} \subseteq D_{a,b}$.
    \end{subproof}
    \begin{claim}\label{xclm:w=Z}
        Both tuples $\bar{w}_{a, W_{a,b}}$ and $\bar{w}_{b, W_{a,b}}$ contain
        exactly the vertices of $Z'_{a,b}$.
    \end{claim}
    \begin{subproof}
      This follows from \Cref{def:exprcolors}
      because for any two vertices $c=[p_c, h_c]$ and $d=[p_d, h_d]$
      satisfying $(s_Q(p_c), h_c) = (s_Q(p_d), h_d)$ and
      $p_c, p_d \in V(B_r(Q, p_a))$, we directly get $h_c=h_d$, and $p_c=p_d$
        (so, $c=d$) since $s_Q$ is a proper coloring of $Q_{3r}$ and
        $p_cp_d \in E(Q_{3r})$ if not equal.
    \end{subproof}
    \begin{claim}\label{xclm:Zsepar}
      The set $Z'_{a,b}$ $r$-separates $a$ and $b$ in~$G$.
    \end{claim}
    \begin{subproof}
        Recall that $a \in V(\phi_{t_i})$ and $b \in V(\psi_h)$.
      We assume that the distance between $a$ and $b$ in $G$
      is at most~$r$; otherwise, the claim is trivial.
        Consider any path $P \subseteq G$ of length at most $r$ between $a$ and $b$.
      Recall that the mapping $[q, h] \mapsto h$ from vertices of $G$ to
      their parameter vertices in the expression is a graph homomorphism,
      that is, it preserves the edge relation.
      We denote by $P_M$ the image of $P$ in $M$ under this homomorphism.

      Let $h_0$ be the $\preceq_M$-minimal vertex of $P_M$.
      Then, $P_M$ certifies that $h_0 \in \wreach{r}{M}{h_a} \cap \wreach{r}{M}{h_b}$.
      It follows from the interpolation property
      that $P_M$ contains a vertex $h'\in X_t$.
        Notably, $h' \not \in Y_{t_i}$
        and $h_0 \preceq_M h'$, so $h_0 \not \in Y_{t_i}$.
      Let $z_0=[p_0, h_0]$ be an arbitrary preimage of $h_0$ in~$G$.
      Furthermore, since the distance in $G$ between
      each of $a$, $b$ and $z_0$ is at most $r$,
      we get that the distance between each of $p_a$, $p_b$
      and $p_0$ in $Q$ is at most $r$, and so
      $p_0 \in V(B_r(Q, p_a)) \cap V(B_r(Q, p_b))$.
        Therefore, by the definition, $z_0 \in Z'_{a,b}$.
    \end{subproof}

%% APPENDIX
%% Duplicated text end
  
    \begin{case}\label{xcase2}
      The edge $xy$ has been added in step~\ref{step:add-edges}\ref{step:add-edges-between-subtrees}.
    \end{case}

    Without loss of generality, we now assume
    that $x, u \in V(\phi_{t_1})$ and $y, v \in V(\phi_{t_2})$.
    Then $\tilde c_{x} = \hat c_{x, t_1}$, $\tilde c_{y} = \hat c_{y, t_2}$,
    $\tilde c_{u} = \hat c_{u, t_1}$, and $\tilde c_{v} = \hat c_{v, t_2}$.

    We observe, as in \Cref{xcase1} again, that $x$ and $u$ (resp. $y$ and $v$)
    have the same type in $G^s$, because
    $\hat c_{x, t_1}(\emptyset) =\hat c_{u, t_1}(\emptyset)$
    (resp. $\hat c_{y, t_2}(\emptyset) =\hat c_{v, t_2}(\emptyset)$)
    is always defined, and the value of $\hat c_{z, t'}(\emptyset)$
    is the type of $z$ for all vertices $z \in V(G)$ and all
    nodes $t' \in V(T)$.
    So, again, $s_Q(p_x) = s_Q(p_u)$ and $s_Q(p_y)=s_Q(p_v)$.

    \medskip
    We again continue the proof with three technical subclaims
    corresponding in order to conditions (I), (II), and (III), as stated before \Cref{xcase1}.

    As in \Cref{xcase1}, this time for both $i\in\{1,2\}$ and any $z=[p_z, h_z] \in V(\phi_{t_i})$, we shortly write
    $Z^i_{z}:=V(B_r(Q,p_z))\times\big(\wreach{r}{M}{h_z}\setminus Y_{t_i}\big)$.
    Given any vertices $a=[p_a, h_a] \in V(\phi_{t_1})$ and
    $b=[p_b, h_b] \in V(\phi_{t_2})$,
    we denote by $Z'_{a,b}$ the set $Z'_{a,b}:=Z^1_a\cap Z^2_b$.
    Let $W_{a,b} := \{(s_Q(p),h) |\> (p,h) \in Z'_{a,b}\}$, and
    let $D_{a,b} \subseteq \mbox{$[d^{3r}+1]$} \times V(M)$ be an arbitrary
    inclusion-maximal set such that both values $\hat c_{a, t_1}(D_{a,b})$ and
    $\hat c_{b, t_2}(D_{a,b})$ are defined.

    \begin{claim}\label{clm:D=W}
      Assume that distance between $p_a$ and $p_b$ is at most $r$ in $Q$.
      Then $D_{a,b} = W_{a,b}$, and so $D_{a,b}$ is uniquely determined.
    \end{claim}
    \begin{subproof}
      Let $(c, h) \in D_{a,b}$ be arbitrary.
      Then, the set $Z_{a, D_{a,b}}$ (cf.~\Cref{def:exprcolors})
      contains a vertex $[p^a, h]$ such that
      $c = s_Q(p^a)$ and $p^a \in V(B_r(Q, p_a))$.
      Furthermore, $h \in \wreach{r}{M}{h_a}\setminus Y_{t_1}$,
      because $\hat c_{a, t_1}$ is the restriction of $\hat c_a$
      to subsets of $[d^{3r}+1]\times(\wreach{r}{M}{h_a} \setminus Y_{t_1})$.
      So, $[p^a, h] \in Z^1_a$.

      We analogously obtain that the set $Z_{b, D_{a,b}}$ 
      contains a vertex $[p^b, h]$ such that
      $c = s_Q(p^b)$, $p^b \in V(B_r(Q, p_b))$, and $[p^b, h] \in Z^2_b$.
      The distance in $Q$ between $p_a$ and $p_b$ is at most $r$,
      and since $s_Q$ is a proper coloring of $Q_{3r}$,
      we get that $p^a = p^b$. Hence, $[p^a, h]=[p^b, h] \in Z'_{a,b}$,
      and so $(c, h) = (s_Q(p^a), h) \in W_{a, b}$.

      Altogether, $D_{a,b} \subseteq W_{a, b}$.
      We observe that both $\hat c_a(W_{a,b})$ and
      $\hat c_b(W_{a,b})$ are defined as witnessed by the vertices in $Z'_{a,b}$.
      Hence, $W_{a,b} \subseteq D_{a,b}$.
    \end{subproof}

    \begin{claim}\label{clm:w=Z}
      Both tuples $\bar{w}_{a, W_{a,b}}$ and $\bar{w}_{b, W_{a,b}}$ contain
      exactly the vertices of $Z'_{a,b}$.
    \end{claim}
    \begin{subproof}
      This follows from \Cref{def:exprcolors}
      because for any two vertices $c=[p_c, h_c]$ and $d=[p_d, h_d]$
      satisfying $(s_Q(p_c), h_c) = (s_Q(p_d), h_d)$ and
      $p_c, p_d \in V(B_r(Q, p_a))$, we directly get $h_c=h_d$, and $p_c=p_d$
      (so, $c=d$) since $s_Q$ is a proper coloring of $Q_{3r}$ and
      $p_cp_d \in E(Q_{3r})$ if not equal.
    \end{subproof}

    \begin{claim}\label{clm:Zsepar}
      The set $Z'_{a,b}$ $r$-separates $a$ and $b$ in~$G$.
    \end{claim}
    \begin{subproof}
      Recall that $a \in V(\phi_{t_1})$ and $b \in V(\phi_{t_2})$.
      We assume that the distance between $a$ and $b$ in $G$
      is at most~$r$; otherwise, the claim is trivial.
      Consider any path $P\subseteq G$ of length at most $r$ between $a$ and $b$.
      Recall that the mapping $[q, h] \mapsto h$ from vertices of $G$ to
      their parameter vertices in the expression is a graph homomorphism,
      that is, it preserves the edge relation.
      We denote by $P_M$ the image of $P$ in $M$ under this homomorphism.

      Let $h_0$ be the $\preceq_M$-minimal vertex of $P_M$.
      Then, $P_M$ certifies that $h_0 \in \wreach{r}{M}{h_a} \cap \wreach{r}{M}{h_b}$.
      It follows from the interpolation property
      that $P_M$ contains a vertex $h'\in X_t$.
      Notably, $h' \not \in Y_{t_1} \cup Y_{t_2}$
      and $h_0 \preceq_M h'$, so $h_0 \not \in Y_{t_1} \cup Y_{t_2}$.
      Let $z_0=[p_0, h_0]$ be an arbitrary preimage of $h_0$ in~$G$.
      Furthermore, since the distance in $G$ between
      each of $a$, $b$ and $z_0$ is at most $r$,
      we get that the distance between each of $p_a$, $p_b$
      and $p_0$ in $Q$ is at most $r$, and so
      $p_0 \in V(B_r(Q, p_a)) \cap V(B_r(Q, p_b))$.
      Therefore, by the definition, $z_0 \in Z'_{a,b}$.
    \end{subproof}
%% APPENDIX end

    Back in \Cref{claim:no-extra-edges}, since
    $\tilde c_{u}(W_{u,v})$ and $\tilde c_{v}(W_{u,v})$
    are defined, we then get that $\tilde c_{x}(W_{u,v})=\tilde c_{u}(W_{u,v})$
    and $\tilde c_{y}(W_{u,v})=\tilde c_{v}(W_{u,v})$ are defined as well.
    Hence, using (I), we get that $W_{u,v} \subseteq W_{x,y}$.
    Similarly, since $\tilde c_{x}(W_{x,y})$ and $\tilde c_{y}(W_{x,y})$
    are defined, we get that $\tilde c_{u}(W_{x,y})=\tilde c_{x}(W_{x,y})$
    and $\tilde c_{v}(W_{x,y})=\tilde c_{y}(W_{x,y})$ are defined as well.
    Hence, $W_{x,y} \subseteq W_{u,v}$.
    Altogether, $W_{x,y} = W_{u,v}$.

    Hence, $\bar{w}_{u, W_{u,v}} = \bar{w}_{v, W_{u,v}}$ (since they are
    the same $\preceq_G$-ordering of the same set $Z'_{u,v}$ by (II)),
    and similarly $\bar{w}_{x, W_{u, v}} = \bar{w}_{x, W_{x, y}} = 
    \bar{w}_{y, W_{x, y}} = \bar{w}_{y, W_{u, v}}$.
    Notably, the tuples $\bar{w}_{u, W_{u, v}} = \bar{w}_{v, W_{u, v}}$
    and $\bar{w}_{x, W_{x, y}} = \bar{w}_{y, W_{x, y}}$
    have the same size equal to $|W_{u, v}|=|W_{x, y}|$
    (see~\Cref{def:exprcolors}).

    Recall that $r \ge 4^q$.
    From (II) and (III), we know that 
    $\bar{w}_{u, W_{u,v}}=\bar{w}_{v, W_{u,v}}$
    $4^q$-separates $u$ and~$v$, and likewise
    $\bar{w}_{x, W_{x,y}}=\bar{w}_{y, W_{x,y}}$
    $4^q$-separates $x$ and~$y$.
    So, we apply \Cref{thm:merge-types} to obtain that
    the type $tp_{q}(G^s, uv)$ of the pair $uv$ (resp.~the type
    $tp_{q}(G^s, xy)$ of $xy$) depends only on the types
    $tp_{g^{\max}}(G^s, \bar{w}_{u, W_{uv}}u) =\tilde c_{u}(W_{uv}) =
    \tilde c_{x}(W_{xy}) = tp_{g^{\max}}(G^s, \bar{w}_{x, W_{xy}}x)$
    and $tp_{g^{\max}}(G^s, \bar{w}_{v, W_{uv}}v) = \tilde c_{v}(W_{uv})
    $ $= \tilde c_{y}(W_{xy}) = tp_{g^{\max}}(G^s, \bar{w}_{x, W_{xy}}x)$;
    hence $tp_{q}(G^s, uv) = tp_{q}(G^s, xy)$.
    In particular, $\xi \in tp_{q}(G^s, uv) = tp_{q}(G^s, xy)$,
    and so $xy \in E(\xi(G))$, a contradiction.

    This finishes the proof of \Cref{claim:no-extra-edges}.
  \end{subproof}

  We have shown that $\phi_t$ satisfies all three invariants.
  However, it remains to show that the number of colors
  is bounded by a function $f$ of $r$, $q$, $d$, and $k$.
  We remark that the whole set $\{\hat c_{v, t} | v \in V(G), t \in V(T)\}$
  has unbounded size. We show that, for each subexpression
  $\phi$ of $\phi_t$, the value of $\phi$
  contains only a bounded number of colors. By doing so,
  we implicitly define the function $f$. Once we have done so,
  we show how to replace the functions $\hat c_{v, t}$
  by numbers from the set $[f(r,q,d,k)]$ so that
  we obtain an expression with the same value
  and bounded total number of colors.

  \begin{claim}\label{claim:count-colors}
    For each subexpression $\phi$ of $\phi_{t}$,
    the value of~$\phi$ features a number of colors at most
    $f(r, q, d, k)$, where the function $f$
    depending only on $r$, $q$, $d$, $k$ is defined
    in~the~proof.
  \end{claim}

  \begin{subproof}
    By the induction assumption, as expressed in \Cref{claim:copy-of-Q},
    the subexpressions $\psi_h$, $\phi_{t_1}$, and $\phi_{t_2}$
    contain at most $f(r, q, d, k)$ (so-far implicit number) of colors. 

    We now bound the number of colors involved when
    creating $\phi_t$ from $\psi_h$, $\phi_{t_1}$, and $\phi_{t_2}$.
    The worst case, over all possible values of these subexpressions,
    defines the function $f(r, q, d, k)$.

    We observe that it suffices to bound the size of
    both unions of domains and codomains of functions $\hat c_a$,
    $\hat c_{a, t}$, $\hat c_{a, t_1}$, and $\hat c_{a, t_2}$.
    Size of the union of the codomains is clearly bounded, since
    there is only bounded number of $g^{\max}$-types of tuples
    of size at most $2\left(d^{3r}+1\right)\cdot{r + k \choose k}$.
    We have already shown that the union of the domains
    of $\hat c_{a, t}$ has bounded size in the proof of \Cref{claim:copy-of-Q},
    and we skip repetition of the argument here.
    We observe that, in order to bound the size of unions
    of the domains of the remaining functions, it suffices
    to bound the size of the set
    $\bigcup_{h \in Y_{t}} \wreach{r}{M}{h} \setminus Y_{t'}$
    for all $t' \in \{t, t_1, t_2\}$.
    Let $t' \in \{t, t_1, t_2\}$ be arbitrary. 
    We show that $\bigcup_{h \in Y_{t'}} \wreach{r}{M}{h} \setminus Y_{t'}
    \subseteq \bigcup_{h \in X_{t'}} \wreach{r}{M}{h}$.
    Let $h_1 \in Y_{t'}$ be arbitrary. Observe that, if
    there is a vertex $h_2 \in \wreach{r}{M}{h} \setminus Y_{t'}$,
    then it follows from \Cref{claim:r-separate-Y_t-from-outside}
    that there is a vertex $p \in X_{t'}$ such that
    $h_2 \in \wreach{r}{M}{p}$.
    Hence,\\
    $\left|\bigcup_{h \in Y_{t'}} \wreach{r}{M}{h} \setminus Y_{t'}\right|
    \le$ $\left|\bigcup_{h \in X_{t'}} \wreach{r}{M}{h}\right| \le
    (k+1) \cdot {r + k \choose k}$. 

    We have shown that each subexpression uses only a bounded
    number of colors. Our upper bounds depend only
    on $r$, $q$, $d$, and $k$, so we have properly defined
    the sought function~$f$.
  \end{subproof}

  We show how to replace the colors $\hat c_{v, t}$
  by numbers from the set $[f(r, q, d, k)]$.
  We do so recursively, top-down, starting at the root $r_0$ of $T$
  and the expression $\phi_{r_0}$.
  Let $\phi$ be arbitrary subexpression of $\phi_{r_0}$ valued $K$. 
  If $\phi = \phi_{r_0}$, then
  we choose an arbitrary injection from the colors used by $K$ to 
  the set $[f(r, q, d, k)]$. It follows
  from \Cref{claim:count-colors} that such injective
  function exists.
  If $\phi$ has a parent subexpression
  $\phi^\uparrow$ valued $K^\uparrow$, then we assume
  that we have already defined an injective function
  $o^\uparrow$ mapping colors used by the
  vertices of $K^\uparrow$ to the numbers from the set $[f(r, q, d, k)]$.
  We define analogous function $o$ for $\phi$, such that
  $o$ and $o^\uparrow$ agree on all common colors.
  That is, let $c$ be arbitrary color used by at least
  one vertex of $K$. If $c$ is also used in $K^\uparrow$,
  then we set $o(c) := o^\uparrow(c)$. Otherwise, we
  have obtained $\phi^\uparrow$ from $\phi$ by recoloring
  $c$ to another color $c^\uparrow$.
  If $c^\uparrow$ is not used in the coloring of $K$,
  we set $o(c) := o^\uparrow(c^\uparrow)$.
  Otherwise, $c^\uparrow$ is used in $K$,
  so $K^\uparrow$ uses at most
  $f(r, q, d, k)-1$ colors (that is, at most one color less
  than~$K$). Since there can be only one such color $c$
  (because \Cref{def:Hexpression} allows only
  one recoloring at a time), we set $o(c)$ to the number
  from the set $[f(r, q, d, k)]$ unused by $o^\uparrow$.

  Finally, in each subexpression, we replace the colors
  by the assigned numbers from the set $[f(r, q, d, k)]$ according
  to the above constructed assignment. By doing so,
  we obtain an expression with the same value
  (except for the color names).

  Since $X_{r_0}=\emptyset$, we get that $Y_{r_0} = V(M)$,
  so $V(Q) \times Y_{r_0} = V(G)$, and it follows
  from the invariant (iii) that the value of
  $\phi_{r_0}$ is $\xi(G)$. Hence, we obtain a
  $(Q_r^\circ, f(r, q, d, k))$-expression for $\xi(G)$, thus finishing the
  proof of \Cref{lem:transductions}.
\end{proof}

Lastly, we derive the two main theorems of this article.

\begin{theorem}\label{thm:transd-admit-prod-str-bound-degree}
  Assume a class $\ca Q$ of bounded-degree graphs, and let
  $\ca C$ be a class of graphs admitting
  $\ca Q$-product structure.
  Let $\tau$ be a first-order transduction
  such that $\tau(\ca C)$ is a class of simple graphs.
  Then, for some integer $r$ depending only on $\tau$, the following holds.
  Denoting by $\ca Q^\circ_r$ the class containing the reflexive closures
  of $r$-th powers of graphs from $\ca Q$,
  the class $\tau(\ca C)$ is a perturbation
  of a class of bounded $\ca Q^\circ_r$-clique-width.

  Furthermore, if $\tau$ is immersive, then $\tau(\ca C)$ itself
  is of bounded $\ca Q^\circ_r$-clique-width.
\end{theorem}
\begin{proof}
  Recall \Cref{thm:pertimmersive} showing
  that every first-order transduction
  $\tau$ is subsumed by a composition of a copying operation $C$,
  an immersive transduction $\tau'$ and a perturbation $P$.
  Since $\tau(\ca C)$ is a class of simple graphs,
  we get that $\tau'$ contains only a single formula $\xi'$.
  Furthermore, it follows from immersiveness that
  $\xi'$ is strongly $r$-local for some $r$
  depending only on $\tau$.

  Without loss of generality, we assume that $\ca C$
  is closed under the operation of adding
  a leaf to an arbitrary vertex. The class $\ca C$
  still admits product structure, since we can simply
  add leaves to the factor $M$ of \Cref{def:prodstruct}.
  \begin{claim}\label{claim:leaf-closed-no-color-copy}
    A class $\ca C$ closed under the operation of adding
    a leaf to an arbitrary vertex does not need coloring
    nor copying.
  \end{claim}
   We include a folklore proof for the sake of completeness.
    \begin{subproof}
      First, we show how to ``simulate'' $k$-copying by
      adding leaves and colors.

      Let $G \in \ca C$ be arbitrary. We denote by $kG$
      its $k$-copy on vertex set $V(kG) = [k] \times V(G)$.
      Let $H$ be a graph obtained from $G$ by adding $k$
      leaves $w_1^v, w_2^v, \ldots, w_k^v$ to each vertex $v \in V(G)$.
      Consider a coloring $c$ in which
      all vertices of $V(G) \subseteq V(H)$ have color
      $0$, and each vertex $w_i^v$ has color $i$.
      Using $c$-colored $H$, we interpret a graph $H^+$
      such that $H^+\left[\bigcup_{v \in V(G)}\{w_1^v, \ldots, w_k^v\}\right]$
      is isomorphic to $kG$ as witnessed by an isomorphism
      mapping vertex $w_i^v \in V(H^+)$ to vertex $[i, v] \in V(kG)$.

      It suffices to observe that $[i, v]$ and $[j, u]$ are adjacent
      in $kG$ if and only if $i=j$ and $uv \in E(G)$, that is,
      if and only if $w_i^v$ and $w_j^u$ have the same color
      $c(w_i^v)=c(w_j^u) > 0$ and there is a path
      $(w_i^v, a, b, w_j^u)$ of length $3$ between $w_i^v$ and $w_j^u$
      (because the only neighbor $a$ of $w_i^v$ is $v$, resp. 
      $b$ of $w_j^u$ is $u$). This is clearly FO expressible.

      Two vertices $w_i^v$ and $w_j^u$ are the
      copy of the same vertex $v=u$ if and only if
      $w_i^v$ and $w_j^u$ share their unique neighbor,
      so we can FO-define the relation $SV$ (cf.
      the definition of the $k$-copy operation).
      Two vertices $w_i^v$ and $w_j^u$ belong to the same copy
      of $G$ if and only if they have the same color
      $c(w_i^v) = c(w_j^u) \ge 1$, so we can FO-define
      the relation $SC$ (cf. definition of $k$-copy operation).

      Thus far, we have shown how to transduce (without copying)
      a $k$-copy of any graph $G \in \ca C$ in another graph
      $H \in \ca C$.

      As the second and final step of this proof, we show how to
      ``simulate'' coloring.

      Let $G \in \ca C$ be an arbitrary graph and let $c: V(G) \to [k]$
      be its arbitrary $k$-coloring (where $k \ge 1$).
      Let $H$ be a graph obtained from $G$ by adding $c(v)+1$
      leaves $w_0^v, w_1^v, \ldots, w_k^v$ to each vertex $v \in V(G)$.
      Observe that, $v \in V(H)$ is a vertex of $V(G) \subseteq V(H)$
      if and only if the degree of $v$ in $H$ is at least 2.
      Furthermore, $v \in V(G) \subseteq V(G)$ has color $c(v)$
      if and only if $v$ has exactly $c(v)+1$ neighbors
      of degree one in $H$. Both of these properties
      are FO-expressible, so we can interpret
      the $c$-colored graph $G$ is the graph $H \in \ca C$.

      Altogether, we get that for any transduction $\tau$
      there is an interpretation $\iota$ such that
      $\tau(\ca C) \subseteq \iota(\ca C)$.
    \end{subproof}

  Using \Cref{claim:leaf-closed-no-color-copy},
  we observe that we can modify $\xi'(x_1, x_2)$
  to obtain a formula $\xi(x_1, x_2)$ such that,
  for every $G' \in C(\ca C)$ and every coloring
  $c$ of $G'$ using the colors referenced in $\xi'$,
  there is a graph $G'' \in \ca C$ such that,
  for every $u, v \in V(G')$, we have
  $G'' \models \xi'(u, v)$ iff $G', c \models \xi(u, v)$.
  
  Hence, it suffices to 
  show, for every strongly local first-order formula $\xi$,
  that the $\xi$-interpretation of a graph $G \in \ca C$
  has bounded $\ca Q^\circ_r$-clique-width.
  Since $G\subseteq Q \boxtimes M$ (\Cref{def:prodstruct})
  for some $Q\in\ca Q$ and some graph $M$ of bounded
  tree-width, we apply \Cref{lem:transductions}
  onto the following objects:
  \begin{itemize}
  \item the 2-coloring $H$ of the graph $Q \boxtimes M$, where the color of
  $v\in V(H)$ is 1 if $v\in V(G)$ and 2 otherwise,
  \item and the formula $\xi_1$ created from $\xi$
  by restricting each quantifier to vertices of color~1 (i.e., belonging to~$G$).
  \end{itemize}

  Doing so, we obtain a 
  $(Q^\circ_r, \ell)$-expression $\phi$ for a graph~containing
  $\xi(G)$ as an induced subgraph,
  where $\ell$ depends on~$\tau$, the degree bound of $\ca Q$ and the
  tree-width of $M$, and~$Q^\circ_r\in\ca Q^\circ_r$.

  Thus, we have shown that $\tau(\ca C)$ is a perturbation
  of the class $\ca D = \{\xi(G) |\> G \in \ca C\}$
  of bounded $\ca Q^\circ_r$-clique-width.
\end{proof}

In the special case of paths, we can strengthen
\Cref{thm:transd-admit-prod-str-bound-degree}.
Let $\ca P^\circ$ denote the class of reflexive closures of all paths.

\begin{theorem}\label{thm:transd-admit-prod-str-paths}
  Let $\ca C$ be a class of graphs admitting product
  structure. Let $\tau$ be a first-order transduction
  such that $\tau(\ca C)$ is a class of simple graphs.
  Then, the class $\tau(\ca C)$ is a perturbation
  of a class of bounded $\ca P^\circ$-clique-width.
\end{theorem}
\begin{proof}
  From \Cref{thm:transd-admit-prod-str-bound-degree}
  we get that $\tau(\ca C)$ is a perturbation
  of a class of bounded $\ca P_r^\circ$-clique-width
  for some constant $r$, where $\ca P_r^\circ$
  is the class of reflexive closures of $r$-th powers~of~paths.

  We show that, given any $(P_r^\circ, \ell)$-expression
  $\psi$ valued graph~$G$, where $P_r^\circ \in \ca P_r^\circ$,
  one can build a $(Q^\circ, 3r\cdot\ell )$-expression
  $\phi$ valued $G$, where $Q^\circ \in \ca P^\circ$.
  
  Let $(p_0, p_1, \ldots, p_n)$
  be the vertices of $P_r^\circ$ in the natural
  order along the underlying path of $P_r^\circ$.
  Observe that contracting each vertex $p_i$ to
  $p_{\left\lfloor \frac{i}{r} \right\rfloor r}$
  creates a reflexive path $Q^\circ$ from $P_r^\circ$.
  We obtain a $(Q^\circ, 3r\cdot\ell)$-expression
  $\phi$ valued $G$ from $\psi$ by
  replacing each parameter vertex $p_i$ by
  $p_{\left\lfloor \frac{i}{r} \right\rfloor r}$,
  and using $3r$-times more colors to emulate
  the behavior of edge addition~in~$\phi$.
  
  More specifically, in each subexpression $\alpha$
  of $\phi$ corresponding to a
  subexpression $\beta$ of $\psi$,
  the color of a vertex $v$ is
  $(c_v, i_v \mod 3r)$, where $i_v$ is the
  index of the parameter vertex $p_{i_v}$ of $v$
  and $c_v$ is the color of $v$ in $\beta$.
  Then, each recoloring operation in $\psi$ corresponds
  to a sequence of recoloring operations in $\phi$,
  disjoint union operations remain unchanged,
  each operation creating a new vertex uses the modified color,
  and each edge addition operation between colors $c^1$
  and $c^2$ is emulated as follows.
  For each pair $i, j$ such that
  $|i - j| \equiv 1 \mod 3r$ or $|i - j| \equiv 0 \mod 3r$,
  we create edges between vertices of colors $(c^1, i)$
  and $(c^2, j)$ provided their parameter vertices
  are adjacent. We observe that, by doing so, 
  we create in $\phi$ exactly the same edges as in $\psi$.
\end{proof}

\begin{remark}
  The proof of \Cref{lem:transductions} is constructive. % in the following sense.
  We can interpret this proof as an FPT-time algorithm for computing
  a $(Q_r^\circ, \ell)$-expression for graphs created by
  a $\xi$-inter\-pretation of a 2-coloring of $Q \boxtimes M$,
  parameterized by $\xi$,
  the tree-width of $M$, and the maximum degree of $Q$;
  provided that $\xi$, $Q$, $M$, and the 2-coloring of $Q \boxtimes M$
  are a part of the input.

  Note, on the other hand, that knowing the 2-coloring of $Q \boxtimes M$ is
  crucial for an efficient computation of our $(Q_r^\circ, \ell)$-expression since,
  in particular, already testing embeddability of a graph in the strong
  product of a path and a graph of bounded tree-width (tree-depth) is
  NP-hard~\cite{DBLP:conf/cccg/BiedlEU23}.

  Consequently, for Theorems~\ref{thm:transd-admit-prod-str-bound-degree}~and~\ref{thm:transd-admit-prod-str-paths},
  we also get analogous FPT-time algorithms for computing
  $(Q_r^\circ,\ell)$- or $(P^\circ,\ell)$-expressions
  of the members of $\tau(\ca C)$, but, importantly, for computing an
  expression for a graph $G'\in\tau(G)$ where $G\in\ca C$,
  the algorithms based on \Cref{lem:transductions} require $G$ and the particular coloring
  used by the transduction $\tau$ on the input.
  
  On the other hand, by analyzing the proof of \Cref{thm:pertimmersive} in
  \cite{DBLP:conf/csl/NesetrilMS22},
  we observe that the ``splitting'' of a transduction $\tau$ into
  a composition of a perturbation, an immersive transduction,
  and copying can be computed
  in FPT time on classes admitting FPT FO model
  checking, such as on considered $\ca C$, provided that the input
  contains $\tau$ and $G \in \ca C$ together with its coloring.
\end{remark}

\section{Backwards translation}
%%%%%%%%%%%%%%%%%%%%%%%%%%%%%%%%%
\label{sec:stabilityback}

In regard of \Cref{thm:transd-admit-prod-str-bound-degree}
it is natural to ask to which extent the converse direction may hold;
that is, whether each class being a perturbation of a class of
bounded $\ca Q^\circ$-clique-width is a transduction of a class admitting
a $\ca Q$-product structure.
This, however, cannot be true in full generality due to the following
case (with $\ca Q^\circ$ containing only the single-vertex loop):
  \begin{theorem}[Ne\v{s}et\v{r}il, Ossona de Mendez, Pilipczuk,
      Rabinovich, and Siebertz \cite{DBLP:conf/soda/NesetrilMPRS21}]
      \label{thm:stablecwtransd}
    If a class of graphs $\ca C$ has bounded clique-width,%
\footnote{\cite{DBLP:conf/soda/NesetrilMPRS21} formulated the statement with
functionally equivalent rank-width.}
    then the following conditions are equivalent:
    \begin{itemize}
      \item $\ca C$ has a stable edge relation,
      \item $\ca C$ is a transduction of a class of bounded tree-width.
    \end{itemize}
  \end{theorem}

On the other hand, adding a suitable stability assumption already
allows to formulate the converse in full strength.
\begin{theorem}\label{thm:stablecw-back}
  Let $\ca Q$ be a class of bounded-degree graphs, let
  $\ca Q^\circ$ be the reflexive closure of $\ca Q$,
  and let $k$ be an integer.
  Let $\ca C$ be a graph class such that, for each $G \in \ca C$,
  there is a $k$-stable $(Q^\circ, k)$-expression $\phi$ for some
  $Q^\circ\in\ca Q^\circ$,
  such that $G$ is a $k$-perturbation of the value of $\phi$.

  Then, there is a class $\ca D$ admitting $\ca Q$-product
  structure and a first-order transduction $\tau$
  such that $\ca C \subseteq \tau(\ca D)$.
\end{theorem}

Before we prove \Cref{thm:stablecw-back}, we remark that
its view of stability is equivalent to the more traditional view as follows:

\begin{claim}\label{claim:equiv-stability}
  Let $\ca C$ be a class of graphs and let $\ca Q$ be
  a class of bounded degree graphs and let $\ca Q^\circ$
  be the reflexive closure of~$\ca Q$.
  Then, the following are equivalent:
  \begin{itemize}
    \item[(A)] There are constants $k$ and $\ell$ such that, for
    each graph $G \in \ca C$, there is a graph
    $Q^\circ \in \ca Q^\circ$ and a $k$-stable
    $(Q^\circ, \ell)$-expression valued $G$.
    \item[(B)] There is a stable bounded clique-width class
    $\ca D$ such that, for each graph $G \in \ca C$,
    there are graphs $H \in \ca D$ and $Q \in \ca Q$
    such that $G$ is an induced subgraph of $Q \boxtimes H$.
  \end{itemize}
\end{claim}
\begin{subproof}
    We show that (A) implies (B).
    Let $Q \in \ca Q$ be a graph and let $Q^\circ \in \ca Q$
    be its reflexive closure.
    Let $\phi$ be a $k$-stable
    $(Q^\circ, \ell)$-expression valued $G$,
    and let $\phi_{depar}$ be its deparameterization valued $G_{depar}$.
    Then, $G_{depar}$ has clique-width at most $\ell$,
    and $G_{depar}$ excludes bi-induced half-graph of order $k$.
    Using the construction from the proof of \Cref{thm:hcwproduct},
    we observe that $G$ is an induced subgraph of $Q \boxtimes G_{depar}$.
    
    Let $\ca D:=\{G_{depar} |\> G \in \ca C\}$ be the class obtained
    from $\ca C$ by the above described process.
    Then, it follows from \Cref{thm:stablecwtransd}
    that $\ca D$ is stable. Furthermore, its 
    clique-width is at most $\ell$, and so $\ca D$ witnesses that
    (B) holds.

    We show that (B) implies (A).
    Let $G \in \ca C$ be arbitrary. Let $H \in \ca D$ and $Q \in \ca Q$
    be graphs satisfying $G \subseteq_i Q \boxtimes H$.
    Let $Q^\circ$ be the reflexive closure of $Q$.
    Let $\ell$ be the clique-width of $\ca D$,
    and let $k'$ be the order of the largest bi-induced
    half-graph in~$\ca D$.

    Let $\phi$ be the $(Q^\circ, \ell)$-expression valued $G$
    constructed in the proof of \Cref{thm:hcwproduct}.
    We observe that, for each $h \in V(H)$,
    there is a subexpression $\phi_h$ valued the copy
    $G[V(Q) \boxtimes \{h\}]$ of $Q$. Furthermore,
    all vertices in $\phi_h$ have the same color and
    form a clique in the value $G_{depar}$ of
    the deparameterization $\phi_{depar}$ of $\phi$,
    so they are twins in $G_{depar}$.
    That is, one can obtain $H$ from $G_{depar}$
    by identifying twins, and so any induced subgraph
    without twins in $G_{depar}$ is also an induced
    subgraph of $H$.

    Suppose that there is a bi-induced half-graph $J$
    in $G_{depar}$ of order $k^+:=2(k'+2)^2$. We obtain a contradiction
    by finding a bi-induced half-graph of order $k'+1$ in $H$:

    Let $A=(a_1, a_2, \ldots, a_{k^+}), B=(b_1, b_2, \ldots, b_{k^+})$
    be the natural bipartition of $J$ with the orderings on $A$ and $B$.
    That is, $a_ib_j$ is an edge if and only if $i \le j$.
    Observe that, within each part $A$ or $B$, there are no twins,
    and so the possible twin pairs within $J$ necessarily form a matching between $A$ and $B$ (not perfect in general).
    Our aim thus is to find a bi-induced half-graph $J'\subseteq J$ of order $k'+1$ which avoids all these twin pairs.

    For $i\in[2k'+4]$, we denote by $A_i:=\{a_j|$ $(i-1)(k'+2)+1\leq j\leq i(k'+2)\}\subseteq A$
    and $B_i:=\{b_j|$ $(i-1)(k'+2)+1\leq j\leq i(k'+2)\}\subseteq B$.
    Next, for $j=1,2,\ldots,k'+1$ in this order, we choose (arbitrarily)
    $a_{i_j}\in A_{2j-1}$ such that $a_{i_j}$ has no twin among the previously chosen vertices of $B$ -- this is possible since $|A_{2j-1}|=k'+2>j$,
    and we analogously choose $b_{i'_j}\in B_{2j+2}$ such that $b_{i'_j}$ has no twin among the previously chosen vertices of $A$.
    Then $J':=J[\bigcup_{j\in[k'+1]} \{a_{i_j}, b_{i'_j}\}]$
    is a bi-induced half-graph of order $k'+1$ in $G_{depar}$.
    Furthermore, $J'$ has no twins, so it is isomorphic to
    an induced subgraph of $H\in\ca D$, a contradiction.
  \end{subproof}

\begin{proof}[Proof of \Cref{thm:stablecw-back}]
  Let $G \in \ca C$ be arbitrary.
  Let $Q_G^\circ \in \ca Q^\circ$, and
  let $\phi_G$ be a $k$-stable $(Q_G^\circ, k)$-expression
  such that $G$ is a $k$-perturbation of the value $G'$
  of $\phi_G$. We denote by $Q_G'\in\ca Q$ the graph obtained from
  $Q_G^\circ$ by removing all loops.
  Let $\psi_G$ be the deparameterization of $\phi_G$.
  We denote by $M_G$ the value of $\psi_G$,
  which is of clique-width at most $k$.
  Since $\phi_G$ is assumed $k$-stable, $M_G$ does not contain
  bi-induced half-graph of order~$k$.

  Consider the class $\ca M = \{M_G |\> G \in \ca C\}$ which,
  by the previous, has stable edge relation.
  So, by \Cref{thm:stablecwtransd},
  there is a first-order transductions $\hat\tau$
  and a class $\hat{\ca M}$ of graphs of bounded
  tree-width such that $\ca M \subseteq \hat\tau(\hat{\ca M})$.

  Let $\ca M^{tw}$ be the closure of $\hat{\ca M}$
  under the operation of adding a leaf to an arbitrary
  vertex. The tree-width of $\ca M^{tw}$ is still bounded,
  and the class $\ca M^{tw}$ does not need copying
  nor coloring. So, there is an interpretation $\tau'$
  such that $\ca M \subseteq \tau'(\ca M^{tw})$.
  Let $\xi$ be the binary formula used in $\tau'$ for the edge relation.

  Notably, there is a graph $M_G^{tw} \in \ca M^{tw}$
  such that $M_G$ is an induced subgraph of $\tau'(M_G^{tw})$.
  In order to simplify our notation,
  we assume that $V(M_G) \subseteq V(M_G^{tw})$ with identity map of vertices,
  and so that for any pair of vertices
  $u, v \in V(M_G)$ we have $uv \in E(M_G)$ if and only if
  $M_G^{tw} \models \xi(u, v)$.

  We denote by $M_G^{uni}$ the graph obtained from $M_G^{tw}$
  by adding a universal vertex $u_G^{uni}$,
  which increases the tree-width by at most one.
  Our core step is to transduce, with a suitable coloring,
  $Q'_G \boxtimes M_G$ from $Q'_G \boxtimes M_G^{uni}$
  using a transduction formula independent of a particular choice of~$G$.

  \begin{claim}\label{clm:gettofactors}
    Let $X=V\big(Q'_G \boxtimes M_G^{uni}\big)$.
    There exist first-order formulas (constructed independently of~$G$) which,
    for a suitable coloring of~$X$, express the following properties
    of any vertex pair $x,y\in X$ over $Q'_G \boxtimes M_G^{uni}$:
    \begin{enumerate}[label={(\roman*)}]
      \item that $x=[q,m]$ and $y=[q',m']$ for $q=q'$ (``same row''),
      \item that $x=[q,m]$ and $y=[q',m']$ for $qq'\in E(Q'_G)$,
      \item assuming that $x=[q,m]$ and $y=[q',m']$ satisfy $q=q'\vee qq'\in
      E(Q'_G)$, the properties that $m=m'$ (``same column'') and that $mm'\in E(M_G^{tw})$.
    \end{enumerate}
  \end{claim}

  \begin{subproof}
    A coloring of a graph $G$ is \emph{proper} if every edge of $G$ receives
    distinct colors for its ends.
    Since the factor $Q'_G$ is of bounded degree and the factor $M_G^{uni}$
    is degenerate thanks to being of bounded tree-width, there exist proper
    colorings $c_1$ and $c_2$ of $Q'_G$ and $M_G^{uni}$, respectively, each
    using a bounded number of colors.
    We color $X$ by the Cartesian product $c=c_1\times c_2$ of the two colorings.
    We proceed to define the formulas expressing (i), (ii), and (iii):

  (i) The formula $\sigma_1$ reads that $c_1(x)=c_1(y)$ and $xy$ is an edge, or that
    there exists a vertex $z$ such that $c_1(x)=c_1(y)=c_1(z)$ and $xz,zy$ are edges.

    Assume $Q'_G \boxtimes M_G^{uni}\models\sigma_1(x,y)$.
    If $xy$ is an edge, but $q\not=q'$, then $qq'\in E(Q'_G)$ by the
    definition of a strong product.
    However, $c_1(x)=c_1(y)$ is in a contradiction with $c_1$ being a proper
    coloring of~$Q'_G$.
    On the other hand, if $xz,zy$ are edges where $z=[q'',m'']$, and $q\not=q'$,
    then, up to symmetry, we get $q\not=q''$.
    As in the previous subcase, we have $c_1(x)=c_1(z)$ and $qq''\in
    E(Q'_G)$, which is again a contradiction to~$c_1$.

    Now assume $q=q'$. Therefore, trivially $c_1(x)=c_1(y)$.
    We choose $z$ as $z=[q,u_G^{uni}]$ where $u_G^{uni}$ is the universal vertex.
    If $z=x$ (up to symmetry), we have $m'u_G^{uni}\in E(M_G^{uni})$, and so
    $zy=xy$ is an edge by the definition of a strong product.
    If $z\not\in\{x,y\}$, then we likewise get that both $xz$ and $zy$ are edges.
    In both cases, $Q'_G \boxtimes M_G^{uni}\models\sigma_1(x,y)$.

  (ii) Using the formula of (i), this formula $\sigma_2$ reads that $x$ and $y$ are not
    from the same row, and there exists $z$ such that $z$ is in the same row
    with $x$ and $yz$ is an edge.

    Assume $Q'_G \boxtimes M_G^{uni}\models\sigma_2(x,y)$.
    Then $q\not=q'$ and $z=[q,m'']$, and $qq'\in E(Q'_G)$ by the
    definition of a strong product since $yz$ is an edge.
    Conversely, assume $qq'\in E(Q'_G)$, which also means $q\not=q'$.
    We choose $z=[q,m']$ (same column as $y$ and same row as~$x$),
    and hence $yz$ is an edge by the definition of a strong product.
    Hence, $Q'_G \boxtimes M_G^{uni}\models\sigma_2(x,y)$.

  (iii) We first, under the assumption that $x=[q,m]$ and $y=[q',m']$
    satisfy $q=q'\vee qq'\in E(Q'_G)$, express that $m=m'$.
    We give a formula $\sigma_3$ reading that
    $x=y$, or $xy$ is an edge and $c_2(x)=c_2(y)$.
    We claim that, under the stated assumptions on $x$ and $y$,~
    $Q'_G \boxtimes M_G^{uni}\models\sigma_3(x,y)$ if and only if $m=m'$.
    Indeed, as in part (i), the third possibility of $mm'\in E(M_G^{tw})$ is
    now excluded by $c_2$ being a proper coloring of $M_G^{uni}$.

    The last formula $\sigma_4$ then reads that $x$ and $y$ are not
    in the same column, $c_2(x),c_2(y)\not=c_2(u_G^{uni})$,
    and there exists $z$ such that $z$ is in the same
    column with $x$ and $yz$ is an edge.
    The conclusion that $Q'_G \boxtimes M_G^{uni}\models\sigma_4(x,y)$
    $\iff$ $mm'\in E(M_G^{tw})$ under the stated assumptions then follows in
    the same way as in part (ii).
  \end{subproof}
%% APPENDIX end

  Observe that the formulas of (iii), with a bit of technical arguments,
  allows us to ``see'' the factor $M_G^{tw}$ within the product
  $Q'_G \boxtimes M_G^{uni}$ and, consequently, to express the result of the
  interpretation $\tau'(M_G^{tw})$.
  The latter, in turn, together with (i) and (ii) define the edge set of 
  $Q'_G \boxtimes M_G$, as desired.
  Finally, we use \Cref{thm:hcwproduct} to argue that
  an induced subgraph of $Q'_G \boxtimes M_G$
  is the value of the $(Q_G^\circ, k)$-expression $\phi_G$ we started with.

    Formally, our proof is finished as follows.
    In view of our assumptions, namely $V(M_G) \subseteq V(M_G^{uni})$,
    it suffices to construct a binary first-order formula $\varrho$ such
    that $Q'_G\boxtimes M_G^{uni}\models \varrho(x,y)$ $\iff$
    $xy\in E\big(Q'_G \boxtimes M_G\big)$.
    Let $x=[q,m]$ and $y=[q',m']$.

    If $q\not=q'$ and $qq'\not\in E(Q'_G)$, which is expressed using
    \Cref{clm:gettofactors} -- formulas (i) and (ii), then $\varrho$ states
    that $xy$ is not an edge. This is correct by the definition of a strong
    product. Hence, for the remaining parts of the formula $\varrho$, we may
    assume that $q=q'$ or $qq'\in E(Q'_G)$ as in
    \Cref{clm:gettofactors}\,(iii).

    Under the latter assumption, we first show how to interpret from the graph
    $Q'_G\boxtimes M_G^{uni}$ a graph $M_{x,y}$ isomorphic to $M_G^{tw}$ such
    that, moreover, $x,y\in V(M_{x,y})$. This is done as follows.
    Let $Y_0=\{q\}\times V(M_G^{tw})$ and $Y=(Y_0\setminus\{[q,m']\})\cup\{y\}$.
    We choose $Y$ as the domain (vertex set) of $M_{x,y}$, which can be
    expressed using the formulas constructed in \Cref{clm:gettofactors} and the unique
    color $c_2(u_G^{uni})$ of the universal vertex $u_G^{uni}$ of $M_G^{uni}$.
    We likewise interpret the edge relation of $M_G^{tw}$ using specifically
    \Cref{clm:gettofactors}\,(iii), which finishes our interpretation~$\iota$.

    Now, we recall the interpretation $\tau'$ such that $M_G$ is an induced
    subgraph of $\tau'(M_G^{tw})$, and the fact that $\tau'$ does not use coloring.
    Hence, we can compose $\tau'$ after $\iota$ (since $\tau'(M_G^{tw})$ is
    isomorphic to $\tau'(M_{x,y})$) to interpret the edge relation of~$M_G$.
    The latter, together with the formulas of \Cref{clm:gettofactors},
    define the sought formula~$\varrho$.

  Finally, since the class $\{Q'_G\boxtimes M_G^{uni} |\> G\in\ca C\}$ admits
  $\ca Q$-product structure by the definition, and the hereditary closure can
  be expressed as a transduction, too, the only remaining
  task is to argue that the value of the $(Q_G^\circ, k)$-expression $\phi_G$ we
  started with is indeed an induced subgraph of $Q'_G \boxtimes M_G$.
  Examining the (short) proof of \Cref{thm:hcwproduct} in
  \cite{DBLP:conf/mfcs/HlinenyJ24}, one may find out that the
  bounded-clique-width factor used there is actually obtained as the value
  of the deparameterization of~$\phi_G$, as we do here with~$M_G$.
  The whole proof is finished.
  \end{proof}

\section{Excluding grids in transductions}
%%%%%%%%%%%%%%%%%%%%%%%%%%%%%%%%%%%%%%%%%%%%%%%%%%%%%%%%%%%%%%%%%%%%%%%
\label{sec:3d-grids}

Our last contribution are two example application of the results of
\Cref{sec:transd-prod-str}; showing that
\begin{itemize}
\item the class of all 3D grids, and
\item the class of certain non-local modifications of 2D grids,
\end{itemize}
are not transducible from any class admitting product structure, and in particular
not from the class of planar graphs.
\\
Recall that $\ca P^\circ$ denotes the class of reflexive closures of paths. 

\begin{theorem}\label{thm:3Dgridpl}
  For any $k$ and $n$ positive integers,
  the $\ca P^\circ$-clique-width of any $k$-perturbation of
  an $n \times n \times n\,$-grid is in $\Omega_k(n)$.
\end{theorem}
\begin{proof}
  Let $G$ be our $n \times n \times n$-grid.
  Let $H$ be any $k$-perturbation of~$G$.
  Let $\ca L = (L_1, L_2, \ldots, L_k)$ be a partition
  of $V(H) = V(G)$ and $\sim$ a symmetric relation on $[k]$
  witnessing that $H$ is a $k$-perturbation of $G$.
  We say that a pair of parts $(L_i, L_j)$ is \emph{flipped}
  if $i \sim j$, otherwise we say that it is \emph{not flipped}.
  
  Consider a subgraph $G'$ of $G$.
  A part $L_i\in\ca L$ is {\em$V(G')$-empty} if $V(G')\cap L_i=\emptyset$
  and {\em$V(G')$-small} if $|V(G')\cap L_i|\leq12$.
  We use the following simple claim:
  \begin{claim}\label{claim:third-subgrid}
    Let $G_1\subseteq_i G$ be a grid of size $m'\times m'\times m'$, where $m'\geq3m''$.
    If a part $L_i\in\ca L$ is $V(G_1)$-small, then there exists an
    induced subgrid $G_2\subseteq_i G_1$ of size $m''\times m''\times m''$ such that 
    $L_i$ is $V(G_2)$-empty.
  \end{claim}
  \begin{subproof}
    At most 12 of $3^3=27$ pairwise disjoint subgrids of $G_1$ of size $m''\times m''\times m''$ may intersect $L_i$.
  \end{subproof}

  We inductively construct a sequence of graphs $G'_0:=G,$ $G'_1, \ldots, G'_\ell$
  such that, for each $i\geq0$, if some part of $\ca L$ is $V(G'_i)$-small but not $V(G'_i)$-empty,
  then let $G'_{i+1}$ be equal to the graph $G_2$ of \Cref{claim:third-subgrid} called with $G_1=G'_i$.
  Then every part of $\ca L$ is either $V(G'_\ell)$-empty, or not $V(G'_\ell)$-small.

  Moreover, since every step $i$ in the construction of $G'_\ell$ adds a new $V(G'_{i+1})$-empty part, we have $\ell\leq k$,
  and so $G'_\ell$ is an induced subgrid of $G$ of size $m\times m\times m$
  where $m\geq \lfloor n/3^k\rfloor$, that is, $m = \Theta_k(n)$.
  We denote by $G':=G'_\ell$ and by $\ca L'$ the restriction of $\ca L$ to $V(G')$.

  We observe that, there is a symmetric relation $\sim'$,
  such that $\ca L'$ and $\sim'$ witness
  that $H':=H[V(G')]$ is a $k'$-perturbation of $G'$ for some $k' \le k$.
  Note that each nonempty part of $\ca L'$ has size at least $13$.

  Observe that the $m \times m \times m$ grid $G'$ has a maximum degree
  at most 6 and the diameter less than $3m$. We show that $H'$
  has diameter less than $9m$, which follows from:
  \begin{claim}\label{clm:dist3x}
    If $u$ and $v$ are neighbors in $G'$, then the $u$--$v$ distance
    in $H'$ is at most $3$.
  \end{claim}
  \begin{subproof}
    Let $L_i$ be the part containing $u$ and let $L_j$ be the part
    containing $v$. If $(L_i, L_j)$ is not flipped, then
    the claim is trivial.
    Hence, we assume that $(L_i, L_j)$ is flipped.
    Since $u$ has at most $6$ $G'$-neighbors in $L_j$,
    all but at most $6$ vertices of $L_j$ are neighbors of $u$
    in $H'$, that is, $u$ has at least $7$ $H'$-neighbors in $L_j$.
    We analogously obtain that $v$ has an $H'$-neighbor $w$
    in $L_i$, and all but at most $6$ vertices of $L_j$ are neighbors of $w$
    in~$H'$.
    Hence, by pigeon-hole principle, $u$ and $w$ have common
    $H'$-neighbor in $L_j$, or $w$ is an $H'$-neighbor of $u$ if~$i=j$, and
    so the distance between $u$ and $v$ in $H'$ is at most $3$.
  \end{subproof}

  Let $\ell$ be the smallest number such that there is a
  reflexive path $P \in \ca P^\circ$ and a
  $(P, \ell)$-expression $\phi$ valued $H$.
  Since $H'$ is an induced subgraph of $H$,
  there is a $(P', \ell)$-expression $\phi'$ valued $H'$
  for some reflexive subpath $P' \in \ca P^\circ$ of $P$.

  We assume that $P'$ is the shortest possible such path.
  Then, the natural homomorphism from $H'$ to $P'$ is surjective,
  so $P'$ has at most $9m$ vertices because
  the diameter of $H'$ is at most $9m-1$.
  Let $p_1, p_2, \ldots, p_{9m}$ 
  be the vertices of $P'$ in the natural order along the path.

  The following claim is folklore, and we include a proof %in the Appendix
  for the sake of completeness:
  \begin{claim}\label{clm:thirds}
    There is a subexpression $\alpha$
    of $\phi'$ such that $\alpha$ is the disjoint union
    of two smaller subexpressions $\alpha_1$ and $\alpha_2$,
    such that $\alpha_1$ contains between
    $\frac{1}{3}m^3$ and $\frac{2}{3}m^3$ vertices.
  \end{claim}
%% APPENDIX
  \begin{subproof}
    We can see $\phi'$ as a rooted binary tree $T'$ in which
    the nodes with two children correspond to the disjoint union operations,
    leafs correspond to the operations creating a new vertex,
    and nodes with one child correspond to remaining operations.
    We find the node corresponding to a sought subexpression $\alpha$
    as the bottommost union node $a$ of $T'$ such that $\alpha$ contains at
    least $\frac{2}{3}m^3$ vertices.
    Then, since $a$ has two children, in one of the two subexpressions
    there has to be at least $\frac{1}{2}\cdot\frac{2}{3}m^3=\frac{1}{3}m^3$
    vertices (and less than $\frac{2}{3}m^3$).
  \end{subproof}
%% APPENDIX end

  Hence, there is a partition $(A, B)$ of $V(H')=V(G')$,
  such that vertices belonging to $A$ are exactly those contained
  within~$\alpha_1$, and $B$ contains the remaining
  vertices. Notably, the size of both 
  $A$ and $B$ is between $\frac{1}{3}m^3$ and $\frac{2}{3}m^3$. 
  From this property, we can already derive the following important claim:

  \begin{claim}\label{claim:linear-matching}
    In the $m \times m \times m\,$-grid $G'$,
    there is an induced matching $M$ between $A$ and $B$
    of size at least $\frac{1}{36\cdot61}m^2$.
  \end{claim}
  \begin{subproof}
    We shortly refer to the parts $A$ and $B$ as to \emph{colors} of $V(G')$;
    that is, $v\in V(G')$ is of color $A$ if $v\in A$, and color $B$ otherwise.
    In the proof we are going to find at least $\frac1{36}m^2$ 2-colored edges in $G'$.
    Then, since the maximum degree in $G'$ is $6$, we can always greedily
    pick a $\frac{1}{61}$-fraction of them to form the sought induced matching~$M$
    -- each picked edge $uw$ ``excludes'' up to $5$ edges incident to $u$,
    up to $5$ edges incident to $w$, and up to $5\cdot(5+5)=50$ more edges incident to the previous edges.

    The grid $G'$ can be clearly vertex-partitioned into $m$ copies of the
    $m \times m\,$-grid, which we further call the \emph{planes} of $G'$,
    and each plane $K$ can be similarly partitioned into copies of the
    $m$-vertex path as into the rows of $K$ or into the columns of $K$, further
    called \emph{lines} of~$G'$.
    A set $X\subseteq V(G')$ is \emph{monochromatic} if $X\subseteq A$ or
    $X\subseteq B$, and $X$ is $\varepsilon$-balanced if at least
    $\varepsilon\cdot|X|$ vertices of it belong to each color $A$ and $B$.
    Observe that every non-monochromatic line of $G'$ contains a 2-colored edge.

    In the first step we show that every $\varepsilon$-balanced,
    $\varepsilon\leq\frac13$, plane $K$ of $G'$ contains $\geq\varepsilon m\>$
    2-colored edges.
    If there are disjoint monochromatic lines in $K$ of opposite colors,
    then all perpendicular lines are non-monochromatic, and hence we already
    have $m$ 2-colored edges.
    Hence, all monochromatic lines, say rows, are of the same color ($A$ or $B$), and
    due to the $\varepsilon$-balanced condition, there are less than
    $(1-\varepsilon)m$ of them.
    So, we get $\geq\varepsilon m\>$ 2-colored edges from the remaining
    non-monochromatic rows.

    In the second step, we finish by examining $\frac16$-balanced planes of~$G'$.
    If there are planes $K_1$ and $K_2$ which are not $\frac16$-balanced, and
    $K_1$ has majority of $A$-color while $K_2$ majority of $B$-color, then
    among the lines perpendicular to them at least $1-\frac26=\frac23$
    fraction of them are non-monochromatic, giving us the needed number
    of 2-colored edges.
    Otherwise, all planes of $G'$ which are not $\frac16$-balanced, have majority in
    the same color, and so, by the property of $V(G')$ being
    $\frac13$-balanced, there are less than $\frac56m$ of them.
    In the remaining at least $\frac16m$ planes which are
    $\frac16$-balanced, we find at least $\frac1{36}m^2$ 2-colored edges by
    the first step argument.
  \end{subproof}

  Let $M$ be the matching of size $|M| \ge \frac{m^2}{36\cdot61}$ from
  \Cref{claim:linear-matching}.
  Recall that there are only $9m$ distinct parameter vertices from~$P'$.
  By the pigeon-hole principle, there is 
  a parameter vertex $p_i$ such that
  at least $\frac{m^2}{36\cdot61}/{9m}=\frac{m}{9\cdot36\cdot61}$
  endpoints of $M$ belonging to $A$ have $p_i$ as their
  parameter vertex. We denote by $A^M_i$
  these vertices, and we denote by $B^M_{i\pm1}$
  the vertices of $B$ adjacent to $A^M_i$ in $M$.

  Observe that the only possible parameter vertices for
  $B^M_{i\pm1}$ are $p_{i-1}$, $p_i$, and $p_{i+1}$.
  Take the expression $\psi$ obtained
  from $\phi'$ by deleting all vertices $v$
  whose parameter vertex $p$ is \emph{not} one of
  $p_{i-1}$, $p_i$, and $p_{i+1}$.
  Let $\alpha'$, $\alpha'_1$, and $\alpha'_2$
  be the subexpressions of $\psi$
  corresponding to $\alpha$, $\alpha_1$, and $\alpha_2$ of \Cref{clm:thirds}.

  Suppose that there are at least $k^2+1$ vertices
  $u_0, u_1, \ldots, u_{k^2} \in A^M_i$
  of the same color in (the value of) $\alpha_1'$.
  Denote by $w_0, w_1, \ldots, w_{k^2} \in B^M_{i\pm1}$ in order their neighbors
  within the matching~$M$.
  By \Cref{def:Hexpression}, vertices of the same color and same parameter
  vertex in $\alpha'_1$ must have the same adjacency to the vertices outside
  $\alpha'_1$; in particular, to the vertices of $B^M_{i\pm1}$.
  Moreover, by pigeon-hole principle, at least two of the vertices
  $u_0, \ldots, u_{k^2}$, say $u_{1}$ and $u_{2}$ up to symmetry,
  belong to the same part of $\ca L'$, and their mates $w_1$ and $w_2$ also
  belong to the same (another) part of $\ca L'$.
  Hence, the vertices $u_{1}$ and $u_{2}$ have the same adjacency to $w_1$
  and $w_2$, which is impossible in a matching or its flip.

  Hence, vertices of $A^M_i$ use at least
  $\frac{1}{9\cdot36\cdot61k^2}\cdot m \in \Omega_k(n)$ distinct colors.
\end{proof}

\begin{corollary}\label{cor:3Dgridsnot}
  Let $\ca C$ be a graph class admitting a classical product structure.
  (Such as, $\ca C$ could be the class of planar graphs, or the class of graphs
  embeddable in a fixed surface.)
  Then, $\ca C$ does not transduce the class of all 3D grids.
\end{corollary}
\begin{proof}
  Suppose otherwise. Then, it follows from \Cref{thm:transd-admit-prod-str-paths}
  that there is a graph class $\ca G$ of bounded $\ca P^\circ$-clique-width
  such that a perturbation of $\ca G$ contains the class of all 3D grids.

  Since perturbations are involutory, we get that $\ca G$ is a
  perturbation of the class of all 3D grids, contradicting
  \Cref{thm:3Dgridpl}.
\end{proof}

For our second example, we define a \emph{pinned grid} of order $n$ as the
graph constructed from the $n^2\times n^2\,$-grid on the vertex set
$\{[i,j]|\> 1\leq i,j\leq n^2\}$ by adding a new apex vertex $w$ adjacent
precisely to the grid vertices of the form $[in,jn]$ for $1\leq i,j\leq n$.
Observe that the class of pinned grids is of locally bounded tree-width and
of bounded twin-width.

\begin{theorem}\label{thm:2Dpingridpl}
  For any $k$ and $n$ positive integers,
  the $\ca P^\circ$-clique-width of any $k$-perturbation of
  the graph which is composed of $n$ disjoint copies of the pinned grid of
  order~$n$,  is in $\Omega_k(n)$.
\end{theorem}
\begin{proof}
We follow similar ideas as in the proof of \Cref{thm:3Dgridpl}.
Let $G=D_1\cup\ldots\cup D_n$ be the disjoint union of copies
$D_1,\ldots,D_n$ of the pinned grid $D$ of order $n$,
and let $H$ be any $k$-perturbation of $G$ witnessed by the partition
$\ca L = (L_1, \ldots, L_k)$ of $V(H) = V(G)$ and relation $\sim$
on~$[k]$.

Let the vertices of $G$ which come as copies of the vertex $w$ of $D$ be
called {apex} vertices of $G$ and the remaining be the grid vertices.
Consider $G'\subseteq G$.
A part $L_i$ of $\ca L$ is \emph{$V(G')$-small} if $L_i$ contains at most $8$
grid vertices of $G'$.
A copy $D_j\subseteq G'$ of $D$ is \emph{$V(G')$-rare} if 
the apex vertex $w_j$ of $D_j$ belongs to a part $w_j\in L_{i'}\in\ca L$
such that $L_{i'}$ contains no other apex vertex of $G'$.
We form $G''\subseteq G$ by iteratively removing from~$G$ the copies $D_j$
of $D$ such that, at the respective iteration $G'\subseteq G$,\, $D_j$
contains a grid vertex of a $V(G')$-small part of $\ca L$ or $D_j$ is $V(G')$-rare.
Observe that in this iterative process, every part of $\ca L$ can become
small at most once, and every removal of a rare copy is counted towards a
different part of $\ca L$.
So, the whole process stops after at most $8k+k$ removal
rounds, removing at most $9k$ copies of $D$ from~$G$.

Consequently, whenever $n>9k$, we have got nonempty $G''\subseteq G$ 
and without loss of generality a pinned grid $D_1\subseteq G''$ of order $n$ such that;
no $V(G'')$-small part of $\ca L$ contains a grid vertex of $D_1$, and $D_1$ is not
$V(G'')$-rare.
Let $H''\subseteq_i H$ denote the corresponding perturbation of $G''$.
Let $w_1$ be the apex vertex of $D_1$.
We get:

\begin{claim}\label{clm:distD1}
Any two vertices of $V(D_1)\setminus\{w_1\}$ have distance in $H''$
at most $18n+4$.
\end{claim}

\begin{subproof}
Let $H^\#$ be $H''$ restricted to the grid vertices of $G''$,
and~$D_1':=D_1-w_1$.
Since the degrees in 2D grids are $\leq4$ and no small part of $\ca L$
intersects $V(D_1')$, we analogously as in
\Cref{clm:dist3x} derive that, for any $u,v\in V(D_1')$
such that $uv\in E(D_1')$, the $u$--$v$ distance in $H^\#$ is at most $3$.

Let $W_1\subseteq V(D_1')$ be the set of the $G''$-neighbors of $w_1$.
Observe that any vertex of $D_1'$ has $D_1'$-distance $\leq2n$ from~$W_1$,
and so $H^\#$-distance $\leq6n$ from~$W_1$ by the previous claim.
Similarly, if we split $W_1$ into two nonempty parts, then the $H^\#$-distance
between the two parts is $\leq6n$, too.
Pick any $u,v\in V(D_1')$.
If none of the edges incident to $w_1$ in $G''$ is flipped in $H''$,
then there is a connection from $u$ to $v$ (via $w_1$) in $H''$ 
of length at most $6n+2+6n=12n+2$.

Further, since $D_1$ is not $V(G'')$-rare, there is an apex vertex
$w_2\in V(G'')$, $w_2\not=w_1$, such that $\{w_1,w_2\}\subseteq L_i\in\ca L$.
Note that $w_2$ is not adjacent to any vertex of $D_1$ in~$G''$.
If all edges incident to $w_1$ in $G''$ are flipped in $H''$,
then $w_2$ is adjacent to all vertices of $W_1$ in $H''$, and we are done as
in the previous paragraph.
Otherwise, every vertex of $W_1$ is in $H''$ adjacent to $w_1$ or $w_2$
(depending on which incident edges are flipped in~$H''$).
Consequently, there is a connection in $H''$ from $u$ via $w_1$ and $w_2$
(up to symmetry) to~$v$ of total length at most $6n+2+6n+2+6n=18n+4$.
\end{subproof}
%% APPENDIX end

Continuing with the proof, let $\ell$ be the smallest number such that there is a
reflexive path $P \in \ca P^\circ$ and a $(P, \ell)$-expression $\phi$ valued $H$.
Since $H''$ is an induced subgraph of $H$,
there is a $(P, \ell)$-expression $\phi''$ valued $H''$.
Consider the $(P, \ell)$-expression $\phi_1$ obtained by restricting $\phi''$
to the grid vertex set $X:=V(D_1)\setminus\{w_1\}$; the value of $\phi_1$ being the
subgraph of~$H''$ induced by~$X$.
We have $|X|=n^2\cdot n^2=n^4$.

Since there is a natural homomorphism from $H''$ to $P$, and by
\Cref{clm:distD1}, the parameter vertices of $\phi_1$ span a subpath of $P$
of length at most $18n+4$.

As in \Cref{clm:thirds}, there is a subexpression~$\alpha$
of $\phi_1$ such that~$\alpha$ is the disjoint union
of two smaller subexpressions $\alpha_1$ and $\alpha_2$,
such that $\alpha_1$ contains between
$\frac{1}{3}n^4$ and $\frac{2}{3}n^4$ vertices.
Hence, there is a partition $(A, B)$ of $X$,
such that vertices belonging to $A$ are exactly those contained
within~$\alpha_1$, and $B$ contains the remaining vertices of~$X$.

\begin{claim}\label{claim:linear-matching-B}
  In the $n^2\times n^2\,$-grid $D_1-w_1$ (where $V(D_1-w_1)=X$)
  with the partition $(A, B)$ of $X$,
  there is an induced matching $M$ between $A$ and $B$
  of size at least $\frac{1}{3\cdot25}n^2$.
\end{claim}
\begin{subproof}
  This is the first step of the proof of \Cref{claim:linear-matching}.
\end{subproof}

Recall that there are only $\leq18n+4$ consecutive parameter vertices
of~$P$ used in~$\phi_1$.
By the pigeon-hole principle, there is a vertex $p_i$ of $P$ such that
at least $\frac{n^2}{3\cdot25}/(18n+4)=\Omega(n)$
endpoints of $M'\subseteq M$ belonging to $A$ have $p_i$ as their parameter vertex.
Now, in the same way as in the proof of \Cref{thm:3Dgridpl},
we conclude that the subexpression $\alpha_1$ requires at least
$\Omega_k(n)$ distinct colors to define (within whole $\phi_1$) the
required perturbation of the selected matching $M'$.
This concludes the proof.
\end{proof}

\Cref{thm:2Dpingridpl} straightforwardly implies via
\Cref{thm:transd-admit-prod-str-paths} the following conclusion analogous to
\Cref{cor:3Dgridsnot}:

\begin{corollary}\label{cor:2Dplusnot}
  Let $\ca C$ be a graph class admitting a classical product structure.
  Then, $\ca C$ does not transduce the class of graphs $\ca D=\{G_n: n\in\mathbb N^+\}$
  where $G_n$ is composed of $n$ disjoint copies of the pinned grid of
  order~$n$.
\qed
\end{corollary}

%% APPENDIX
  \begin{toappendix}
In case connectivity of the example as in \Cref{thm:2Dpingridpl} is
desirable, we add a simple modification of it.
Let a \emph{pinned multigrid} of order $n$ be any connected graph
$G^+_n$ such that $G^+_n$ induces $n$ disjoint copies of the pinned
grid of order~$n$;
such as, $G^+_n$ can be simply constructed from the $(n^3+n-1)\times n^2\,$-grid 
by suitably adding $n$ apex vertices.
\Cref{thm:2Dpingridpl} then straightforwardly implies via
\Cref{thm:transd-admit-prod-str-paths} the following conclusion analogous to
\Cref{cor:3Dgridsnot}:

\begin{corollary}\label{cor:2DplusnotII}
  Let $\ca D$ be the class of graphs $G_n$ for $n\in\mathbb N^+$
  where $G_n$ is composed of $n$ disjoint copies of the pinned grid of
  order~$n$,
  and let $\ca D^+$ be a class containing a pinned multigrid of order
  $n$ for each $n\in\mathbb N^+$.
  If $\ca C$ is a graph class admitting a classical product structure,
  then neither of $\ca D$ and $\ca D^+$ is transducible from~$\ca C$.
\qed
\end{corollary}
  \end{toappendix}
%% APPENDIX end

As noted in the introduction, Gajarsk{\'{y}}, Pilipczuk and Pokr{\'{y}}vka
\cite{DBLP:journals/corr/abs-2501-07558} give an alternative proof of 
\Cref{cor:3Dgridsnot} using the concept of so-called \emph{slice decomposition}.
It can be shown that every class of bounded $\ca P^\circ$-clique-width
admits a slice decomposition, but the converse is not true in general.
For example, the class $\ca D$ of \Cref{cor:2Dplusnot} admits a slice decomposition.

\section{Concluding Remarks}
%%%%%%%%%%%%%%%%%%%%%%%%%%%%%%%%%%%%%%%%%%%%%%%%%%%%%%%%%%%%%%%%%%%%%%%
\label{sec:conclu}

In view of \Cref{thm:transd-admit-prod-str-paths} and
\Cref{thm:stablecw-back}, one comes with the following natural question.
Assume $\ca C$ is a monadically stable class admitting a product structure, such as the
class of planar graphs, and $\tau$ is a first-order transduction.
Then $\tau(\ca C)$ is stable by the assumption, and up to bounded
perturbations, $\tau(\ca C)$ is of bounded $\ca P^\circ$-clique-width $\leq\ell$.
(Recall that $\ca P^\circ$ denotes the class of reflexive paths.)
Does this mean that the perturbed graphs of $\tau(\ca C)$ always have stable 
$(P^\circ, \ell)$-expressions?
This may not be an easy question since, in an analogous instance, 
it is known that there are classes of bounded-degree planar graphs which do
not admit a product structure with bounded maximum
degrees of the factors~\cite{DBLP:journals/combinatorics/DujmovicJMMW24}.

While we are still investigating this question, we propose:
\begin{conjecture}\label{conj:expression-stable}
  There is a function $g: \mathbb{N} \to \mathbb{N}$
  such that each $(\ca Q_r^\circ, \ell)$-expressions constructed in
  the proof of \Cref{lem:transductions} is $g(\ell)$-stable.
\end{conjecture}

Stepping further, we define that a class $\ca C$ is of 
\emph{bounded~ex\-pression-stable} $\ca H$-clique-width if there is $k$ such
that every graph of $\ca C$ is the value of a $k$-stable $(H,k)$-expression
for~$H\in\ca H$.
A positive answer to \Cref{conj:expression-stable} would, together with
Theorems \ref{thm:transd-admit-prod-str-paths} and
\ref{thm:stablecw-back}, imply that classes first-order transducible from
classes admitting a product structure are exactly perturba\-tions of
those of bounded expression-stable \mbox{$\ca P^\circ$-clique-width}.

Previous would also imply validity of the next conjecture:
\begin{conjecture}\label{conj:preserving}
  The property of a graph class -- to be a perturbation of a class of
  bounded expression-stable $\ca P^\circ$-clique-width,
  is preserved under taking first-order transductions.
\end{conjecture}

\Cref{conj:preserving} actually formulates a strengthening of
\Cref{thm:transd-admit-prod-str-paths} in two ways;
not only the conclusion adds expression-stable $\ca P^\circ$-clique-width,
but also the assumption on the class $\ca C$ is weakened from admitting a product
structure to having bounded expression-stable $\ca P^\circ$-clique-width.
An analogous strengthening -- to assume a class $\ca C$ of bounded
expression-stable $\ca Q^\circ$-clique-width, would also be interesting in
\Cref{thm:transd-admit-prod-str-bound-degree} for arbitrary bounded-degree
class~$\ca Q$.

Regarding further possibilities to strengthen
\Cref{thm:transd-admit-prod-str-bound-degree}, we repeat that the assumption
of having a class $\ca Q$ of bounded maximum degree indeed is unavoidable;
\cite{DBLP:conf/mfcs/HlinenyJ24}
the class of strong products of two stars is monadically independent.

\smallskip
We finish with two questions and a remark (\Cref{prop:sparse-admit-prod-str-paths})
reaching slightly beyond the scope of the results here.

Firstly, we say that a class $\ca C$ is of
\emph{bounded expression-sparse} $\ca H$-clique-width if there is $k$ such
that every graph of $\ca C$ is the value of an $(H,k)$-expression $\phi$
for~$H\in\ca H$ such that the value of the deparameterization $\phi_1$ of $\phi$ is $K_{k,k}$-free.
% (note that this is equivalent to claiming that the value of $\phi_1$ is of bounded tree-width).

\begin{question}\label{que:expression-sparse}
Is it true that if a graph class $\ca C$ is of bounded $\ca
P^\circ$-clique-width and $\ca C$ excludes bi-induced $K_{t,t}$
for some~$t$, then $\ca C$ is also
of bounded expression-sparse $\ca P^\circ$-clique-width?
\end{question}
Note that the value of the deparameterization $\phi_1$ being $K_{k,k}$-free
is equivalent to this value being of bounded tree-width.
Hence, in \Cref{que:expression-sparse},
for a class $\ca C$ being of bounded expression-sparse 
$\ca P^\circ$-clique-width is equivalent (cf.~\Cref{thm:hcwproduct}) to 
$\ca C$ admitting product structure.
Being the answer to \Cref{que:expression-sparse} yes,
we would get kind of a `sparse fragment' of \Cref{thm:transd-admit-prod-str-paths}
with the conclusion featuring a product structure as well.

Similarly as with the previous question about expression-stable 
$\ca P^\circ$-clique-width, an answer to \Cref{que:expression-sparse} is not
immediate (cf.~\cite{DBLP:journals/combinatorics/DujmovicJMMW24}), 
and we are not even sure whether this answer should be yes or no.
Nevertheless, a very new result of 
Gajarsk\'y, G{\l}adkowski, Jedelsk\'y, Pilipczuk and
Toru\'nczyk~\cite{DBLP:journals/corr/abs-2505.15655}
allows us to derive a `sparse fragment' of \Cref{thm:transd-admit-prod-str-paths}
in a different way:

\begin{proposition}\label{prop:sparse-admit-prod-str-paths}
  Let $\ca C$ be a class of graphs admitting product
  structure. Let $\tau$ be a first-order transduction
  such that $\tau(\ca C)$ is a class of simple $K_{t,t}$-free
  graphs for some fixed integer~$t$.
  Then, the class $\tau(\ca C)$ ``nearly admits'' a product structure,
  meaning that there is an integer $k$ such that every graph of $\tau(\ca C)$,
  after deleting at most $k$ of its vertices,
  is a subgraph of the strong product of a path and a graph of
  tree-width at most~$k$.
\end{proposition}
\begin{proof}
Consider a graph $G$ and a collection $\ca A\subseteq2^{V(G)}$ of vertex
subsets such that, for each $A\in\ca A$, the induced subgraph $G[A]$ is connected.
Then $\ca A$ is called a {congestion-$c$ depth-$d$ minor model} of a
graph $H$ on $V(H)=\ca A$ if:
(i) every induced subgraph $G[A]$, $A\in\ca A$ is of radius~$\leq d$,
(ii) for every vertex $v\in V(G)$, at most $c$ sets of $\ca A$ contain~$v$,
and (iii) for every edge $AB\in E(H)$ (where $A,B\in\ca A$), the sets $A$
and $B$ \emph{touch} in $G$ -- meaning that $A\cap B\not=\emptyset$ or some
two vertices $a\in A$ and $b\in B$ are adjacent in~$G$.
A graph $G_1$ is a \emph{congestion-$c$ depth-$d$ minor} of $G$ if
there exists a congestion-$c$ depth-$d$ minor model in $G$ isomorphic to~$G_1$.

By our assumptions,
the class $\ca C$ (as admitting product structure) is of bounded expansion
and the class $\tau(\ca C)$ is $K_{t,t}$-free.
Therefore, by Gajarsk\'y, G{\l}adkowski, Jedelsk\'y, Pilipczuk and
Toru\'nczyk~\cite{DBLP:journals/corr/abs-2505.15655}, the following holds;
there exists $c\in\mathbb N$ such that~$\tau(\ca C)$ is contained in the
class of congestion-$c$ depth-$c$ minors of the class $\ca C^\bullet$,
where $\ca C^\bullet$ is the class obtained by adding a universal vertex to
every graph of~$\ca C$.

Let $H\in\tau(\ca C)$ which is a congestion-$c$ depth-$c$ minor 
of $G^\bullet\in\ca C^\bullet$.
Let $G\subseteq G^\bullet$ be the graph of $\ca C$ from which $G^\bullet$ 
was constructed by adding a universal vertex~$u$.
Let $X\subseteq V(H)$ where $|X|\leq c$ be the set of vertices of $H$
whose sets in the minor model contain~$u$.
We are going to show that $H_1:=H-X$ is a subgraph of the strong product of a path 
and a graph of tree-width at most~$k$ ($k$ independent of~$H$).

First, $H_1$ is a congestion-$c$ depth-$c$ minor of~$G\in\ca C$,
and $G\subseteq P\boxtimes M$ where $P$ is a path and $M$ is of tree-width~$\leq m$.
Then, clearly, $H_1$ is a depth-$c$ minor 
of~$G\boxtimes K_c\subseteq P\boxtimes M\boxtimes K_c$.
% and the tree-width of $M\boxtimes K_c$ is at most $c(m+1)$.

Second, Hickingbotham and Wood~\cite{DBLP:journals/siamdm/HickingbothamW24}
proved (their result was actually more general)
that depth-$c$ minors of graphs of the form $P\boxtimes M\boxtimes K_c$
are subgraphs of $P\boxtimes M'$ where $M'$ is of tree-width bounded 
only in terms of our $m$ and $c$.
This concludes the proof.
\end{proof}

Lastly, inspired by \cite{DBLP:journals/tocl/GajarskyHOLR20},
we conclude with an algorithmic question:
\begin{question}
Let $\ca C$ be the class of planar graphs (or a class admitting a product
structure), and $\tau$ be a first-order transduction.
Can we find a first-order transduction $\tau'$ (of $\ca C$) subsuming $\tau$
such that, for every graph $G\in\tau(\ca C)$ there is 
an efficiently computable preimage $G'\in\ca C$ such that $G \in \tau'(G')$?
\end{question}

\bibliography{Hcw}

\end{document}